\documentclass[aps,prx,amsmath,amssymb,twocolumn,10pt,superscriptaddress, a4paper]{revtex4-2}

\usepackage{lmodern}   
\usepackage[T1]{fontenc}
\usepackage[utf8]{inputenc}

\usepackage{enumitem}

\usepackage{mathrsfs}

\usepackage{amsfonts}
\usepackage{amsmath,amssymb,amsthm}
\usepackage{graphicx}
\usepackage{bbm}
\usepackage{hyperref}
\usepackage{mathtools}
\usepackage{physics}
\usepackage{upgreek}
\usepackage{dsfont}
\usepackage[capitalize]{cleveref}
\usepackage{xcolor}

\usepackage{changes}

\usepackage{geometry}
\newtheorem{theorem}{Theorem}[section]
\newtheorem{lemma}[theorem]{Lemma}

\newtheorem{corollary}[theorem]{Corollary}

\newtheorem{assumptions}{Assumptions}

\theoremstyle{definition}
\newtheorem{definition}[theorem]{Definition}
\newtheorem*{theorem*}{Theorem}
\newtheorem*{corollary*}{Corollary}

\DeclareMathOperator{\C}{\mathbb{C}}

\DeclareMathOperator{\N}{\mathbb{N}}

\DeclareMathOperator{\Z}{\mathbb{Z}}

\DeclareMathOperator{\cB}{\mathcal{B}}

\DeclareMathOperator{\cH}{\mathcal{H}}

\DeclareMathOperator{\cO}{\mathcal{O}}

\DeclareMathOperator{\1}{\mathds{1}}

\newcommand{\vertiii}[1]{{\left\vert\kern-0.25ex\left\vert\kern-0.25ex\left\vert #1 
		\right\vert\kern-0.25ex\right\vert\kern-0.25ex\right\vert}}

\newcommand{\pdftitle}{Stability of thermal equilibrium in long-range quantum systems}

\hypersetup{%
	colorlinks,%
	linkcolor={red!50!black},%
	citecolor={blue!60!black},%
	urlcolor={blue!80!black},%
	unicode,%
	pdftitle={\pdftitle}, 
	pdfauthor={Tim M\"{o}bus, Jorge S\'{a}nchez-Segovia, \'{A}lvaro M. Alhambra, \'{A}ngela Capel},
	pdfduplex = DuplexFlipLongEdge,%
	pdflang = en,%
}

\usepackage{subcaption}
\usepackage{caption}

\usepackage{ragged2e}
\justifying

\usepackage[
activate={true,nocompatibility},
tracking=true,
kerning=true,
]{microtype}
\microtypecontext{spacing=nonfrench}

\begin{document}
	
	\title{\pdftitle}
	
	\author{\begingroup
		\hypersetup{urlcolor=black}
		\href{https://orcid.org/0000-0003-0471-2745}{Tim M\"{o}bus
			\endgroup}
	}
	\email{moebustim@gmail.com}
	\affiliation{Department of Mathematics, University of T\"ubingen, 72076 T\"ubingen, Germany}
	
	\author{\begingroup
		\hypersetup{urlcolor=black}
		\href{https://orcid.org/0009-0002-1274-4747}{Jorge~S\'anchez-Segovia
			\endgroup}
	}
	\email{jorge.sanchezsegovia@estudiante.uam.es}
	\affiliation{Instituto de F\'isica Te\'orica UAM/CSIC, C.~Nicol\'as Cabrera 13-15, Cantoblanco, 28049 Madrid, Spain}
	
	\author{\begingroup
		\hypersetup{urlcolor=black}
		\href{https://orcid.org/0000-0002-5889-4022}{\'Alvaro~M.~Alhambra
			\endgroup}
	}
	\email{alvaro.alhambra@csic.es}
	\affiliation{Instituto de F\'isica Te\'orica UAM/CSIC, C.~Nicol\'as Cabrera 13-15, Cantoblanco, 28049 Madrid, Spain}
	
	\author{\begingroup
		\hypersetup{urlcolor=black}
		\href{https://orcid.org/0000-0001-6713-6760}{{\'A}ngela Capel
			\endgroup}
	}
	\email{ac2722@cam.ac.uk}
	\affiliation{Department of Mathematics, University of T\"ubingen, 72076 T\"ubingen, Germany}
	\affiliation{Department of Applied Mathematics and Theoretical Physics, University of Cambridge, United Kingdom}
	
	\begin{abstract}
		\setlength{\parindent}{0pt}
		Experimental realizations of spin models are irremediably prone to errors, which can propagate through the system corrupting experimental signals. We study how such errors affect the measurement of local observables in systems with long-range interactions, where perturbations can spread more rapidly. Specifically, we focus on the stability of thermal equilibrium and investigate its relation to the correlation structure of the system, both analytically and numerically. As a main result, we prove that the stability of local expectation values follows from the decay of correlations on the Gibbs state and the Lieb-Robinson bound. We also provide numerical evidence that this stability extends to an even larger regime of interacting long-range systems. Our results support the robustness of analog simulation platforms for long-range models, and provide further evidence that computing physical quantities of interest can be significantly easier than performing arbitrary quantum computations.
	\end{abstract}
	\maketitle
	
	\section{Introduction}
	A variety of experimental platforms are capable of probing models of interacting quantum particles. In many of these systems, the interactions can be tuned to have a long-range character, decaying as an inverse polynomial of the distance $d^{-\alpha}$ \cite{reviewLR_RevModPhys.95.035002,javier2025engineering}. This includes several types of analog simulation experiments, which are often considered among the leading candidates for achieving a ``useful quantum advantage'' \cite{Daley2022,LiuAdvantage2025}.
	
	The idea behind them is to mimic complex quantum phenomena, enabling the study of many-body physics in regimes beyond classical computation and without the need for fault-tolerant quantum computers. They promise to impact different areas such as condensed matter physics \cite{expCondensedMatter}, lattice gauge theories \cite{expHighEnergy1_Aidelsburger2022}, and quantum chemistry \cite{exp_qchemistry1_Navickas_2025}. Although these simulators are not necessarily universal \cite{Georgescu2014}, they can often prepare ground or thermal states, probe dynamics, and allow for the measurement of key observables \cite{Altman2021}.
	
	All experimental platforms, and likely any other physical realization of many-body models, are bound to contain imperfections in the design. These often take the form of ``coherent errors'', which can be understood as perturbations to the original Hamiltonian \cite{Sarovar2017,Trivedi2024}:
	\begin{equation*}
		H \rightarrow H + \varepsilon V,
	\end{equation*}
	with some other Hamiltonian $V$. It is therefore of interest to understand how large systems respond to such perturbations, and in particular, to explore the conditions under which they remain stable. 
	
	A deeper understanding of this stability is of fundamental interest, to understand how physical features such as the absence of long-range order and the clustering of correlations relate to it. It is also important from the perspective of analog quantum simulation, in order to assess the impact of errors on the measured signals and to establish guarantees regarding their reliability.
	
	These questions have been studied in a variety of scenarios \cite{Sarovar2017,Marthaler2017,Flannigan_2022,Chertkov.2024}. The work \cite{Trivedi2024} studied the stability of dynamics and equilibrium, both at zero and non-zero temperatures, focusing in finite-range spin and fermionic Hamiltonians. Other results include \cite{cai2023randomerrors}, which considered the effect of stochastic errors; \cite{kashyap2024quantumadvantage}, which focused on open analog quantum simulation; and \cite{Poggi2020}, which considered the stability of typical analog dynamics. 
	
	Most of these results focus on finite-range systems. However, the picture becomes more complicated when long-range interactions are present. These can introduce additional collective or long-range effects that can hinder stability. For certain questions, these effects can easily be controlled. For instance, the arguments in \cite{Trivedi2024} show that the stability of dynamics and gapped ground states essentially follows from the existence of a Lieb-Robinson bound, which is relatively well understood in long-range systems both from the analytical \cite{Eisert_2013,Kuwahara.2020lrb,Tran2019,Tran2020,Chen2019,Tran2021,Cevolani_2016,PRXQuantum.4.020348} and numerical \cite{Hauke_2013,jan_PhysRevResearch.3.L012022} point of view. However, the stability of thermal equilibrium is more subtle, and the standard proofs for finite-range systems do not directly apply to long-range systems \cite{Brandao2019,capel2023lppl,Rouze2024Learning,Rakovszky2024}. The effect of local perturbations on local properties of gapped ground states is discussed in \cite{PRXQuantum.4.020348}, and the stability of Gibbs states on one-dimensional systems in \cite{capel2023lppl}; however, the proof techniques used in these works differ substantially and do not apply in our setting. This is unfortunate, since current simulation experiments already study questions related to thermal equilibrium in general long-range models \cite{Neyenhuis2017,Joshi_2022,Schuckert2025}, and there are quantum algorithms that can efficiently sample long-range Gibbs states at high temperatures \cite{rouze2024optimalquantum}.
	
	In this work, we address this by analytically proving the stability of thermal long-range systems in arbitrary dimensions. In particular, we demonstrate that this stability follows from two well-established properties: the decay of correlations on the Gibbs state and the Lieb-Robinson bound. 
	
	We also study these questions numerically by simulating a long-range spin chain with tensor networks, providing evidence that the stability likely holds in a wider regime of interaction strengths. Additionally, we study the related notion of \emph{local indistinguishability} in long-range models, and show how all the different notions of clustering are weakly equivalent.
	
	The proofs, on a high level, involve a perturbative analysis of matrix exponentials in terms of an operator that, due to locality, is of a constrained radius. This, together with clustering properties, yields the stability results \cite{Trivedi2024,capel2023lppl}. These steps are standard for short-range systems. However, for long-range systems, the effect of the perturbations is significantly larger, so that a straightforward translation of previous arguments yields useless bounds. Here we circumvent this difficulty by combining a careful analysis of the matrix exponential from \cite{Kliesch.2014} with the Lieb-Robinson bound.
	
	\section{Preliminaries}
	\label{sec:preliminaries}
	Let $(V,E)$ be a possibly infinite graph and $\Lambda \subset V$ be a finite subset of it, $\cH_d$ a $d$-dimensional Hilbert space and $\cH_A$ for $A\subset \Lambda$ the tensor-product space $\cH_A = \cH_d \otimes \cdots \otimes \cH_d$ of $|A|$ copies of $\cH_d$ indexed by the elements in $A$. When considering the whole system $\Lambda$, we shorten the notation to $\cH_\Lambda=\cH$. The space of bounded operators on $\cH$ is $\cB(\cH)$, and the operator and trace norms are $\|\cdot\|$, $\|\cdot\|_1 \coloneqq \tr[|\cdot|]$, respectively.  A locally supported operator is of the form $O_A \otimes \1_{A^c}$, or simply $O_A$, with the complement $A^c = \Lambda \setminus A$, $O_A \in \cB(\cH_A)$ and the identity acting on $\cH_{A^c}$. Given a finite lattice $\Lambda \subset \mathbb{Z}^D$ with $|\Lambda| = N$ sites, we define the Hamiltonian $H=H^\dagger$ in $k$-local form by
	\begin{equation}\label{eq:k-form-hamiltonian}
		H = \sum_{Z \,:\, |Z| \leq k} h_Z \,,
	\end{equation}
	where $Z \subset \Lambda$ is a set of interaction sites. We assume that the Hamiltonian is long-range, i.e.,
	\begin{equation}\label{eq:def-long-range}
		J_{j,j'} = \sum_{Z \,:\, \{j,j'\} \subset Z} \|h_Z\| \leq \frac{g}{(1 + d(j,j'))^\alpha} \, ,
	\end{equation}
	where $d(j, j')$ denotes the graph distance between sites $j$ and $j'$ on $\Lambda$, $g\geq0$ is a constant and $\alpha>0$ is the index of the power-law decay of the interactions. It is also useful to define the quantity $u=\max_{j'}{\sum_j\frac{2^\alpha}{(1+d(j,j'))^\alpha}}$, which is constant as long as $\alpha > D$. Note that the graph distance between two sets $A$ and $B$ is defined by $d(A,B)\coloneqq \min_{j\in A,\,j'\in B}d(j,j')$. 
	
	The Hamiltonian defines the Gibbs state $\rho_\beta[H] = e^{-\beta H} / \tr[e^{-\beta H}]$ with inverse temperature $\beta \geq 0$ or $\rho_\beta$ if the Hamiltonian is known from the context. To measure the correlations in $\rho_\beta$, we use the covariance
	\begin{equation}\label{eq:def-correlation-function}
		\mathrm{Cov}_{\!\rho_\beta}\!(\hspace{-0.2ex}O_A,O_B\hspace{-0.2ex})\!\coloneqq\!\tr[\rho_\beta O_A O_B]-\tr[\rho_\beta O_A]\!\tr[\rho_\beta O_B]
	\end{equation}
	defined for observables $O_A$ and $O_B$ supported in $A,B \subset \Lambda$. A key assumption in our analysis is \textit{decay of correlations}, which quantifies how the covariance decreases as the distance between $A$ and $B$ increases.
	
	\begin{figure}[t!]
		\begin{center}
			\includegraphics[scale=0.25]{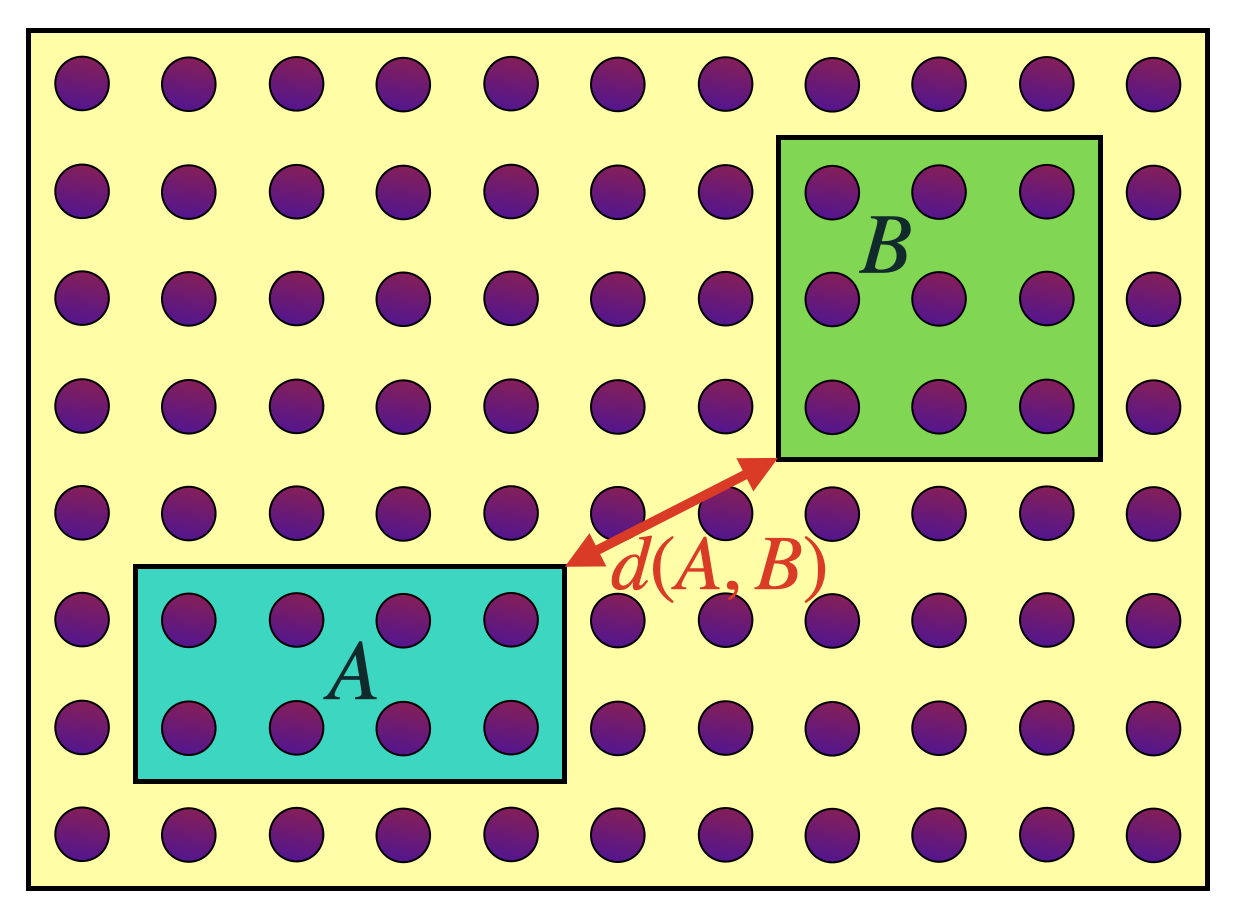}
			\caption{\justifying Two disjoint subregions $A$ and $B$ of a square lattice $\Lambda$. }
			\label{fig1}
		\end{center}
	\end{figure}  
	
	\begin{definition}[Decay of correlations]\label{def:decay-correlation}
		Given an inverse temperature $\beta\geq 0$, and any $A, B \subset \Lambda$ as in Fig.~\ref{fig1}, the Gibbs state $\rho_\beta$ satisfies decay of correlations if there is a constant $K\geq0$ so that
		\begin{equation}
			\mathrm{Cov}_{\rho_\beta}(O_A, O_B) \leq K \|O_A\| \|O_B\| \frac{|A|\,|B| e^{(|A|+|B|)/k}}{(1+d(A, B))^\alpha},
			\label{eq::def_decay_of_correlations}
		\end{equation}
		for all observables $O_A \in \cB(\cH_A) $ and $O_B \in \cB(\cH_B) $.
	\end{definition}
	This functional form is motivated by the recent result in \cite[Thm.~2]{kim2024arealawslongrange}, which demonstrates it under the assumption that $\alpha > D$ and $\beta \leq \beta^*= \frac{1}{8 u g k}$. For the decay of correlations in 1D systems at all temperatures (with a weaker decay exponent $\alpha-2$), see \cite{kimura2024decaycorrelations1D}.
	
	\subsection{Technical tools on locality}\label{sec:locality}
	For our proofs, we assume that the long-range Hamiltonians satisfy a Lieb-Robinson bound of the form
	\begin{equation}\label{eq:def-lieb-robinson-bound-prelim}
		\|[O_A(t),O_B]\|\leq \kappa_{LR}\|O_A\|\|O_B\|\frac{e^{v|t|}-1}{(1+d(A,B))^{\alpha}},
	\end{equation}
	for any $O_A \in \cB(\cH_A) $ and $O_B \in \cB(\cH_B) $, with $v= 2gu$, and $\kappa_{LR}$ a constant depending on $\alpha, g$ and $D$, regions $A,B \subset \mathbb{Z}^D$ and with $O_A(t)\coloneqq e^{itH}O_Ae^{-itH}$. This bound was proven in \cite{hastings2006gapdecaycorrelations}, and it holds across the entire weak long-range regime $\alpha > D$.
	
	More recent results \cite{Tran2019,Tran2020,Chen2019,Tran2021,Kuwahara.2020lrb} prove bounds with improved time dependence but stronger constraints on $\alpha$. In particular, \cite{Kuwahara.2020lrb} proves that for $\alpha\geq2D+1$,
	\begin{equation}\label{eq:def-lieb-robinson-bound-prelim-2d+1}
		\|[O_A(t),O_B]\|\leq \kappa_{LR,1}\|O_A\|\|O_B\|\frac{p(|t|)}{(d(A,B)\!-\!v|t|)^{\alpha}},
	\end{equation}
	where $p(|t|)\coloneqq|t|^{2D+1}$. For the regime $2D < \alpha < 2D+1$, there also exists a Lieb-Robinson bound with a slower polynomial decay \cite{Tran2021}
	
		\begin{align}\label{eq:def-lieb-robinson-bound-prelim-2d}
			\|[O_A(t),O_B]\|\leq& \|O_A\|\|O_B\| \nonumber\\
			& \cross \Bigl[ \kappa_{LR,2} \left(\frac{t}{d(A,B)^{\alpha-2D+\epsilon}}\right)^{\frac{\alpha-D}{\alpha-2D} - \frac{\epsilon}{2}} \nonumber\\
			&\qquad+ \kappa_{LR,3} \left(\frac{t}{d(A,B)^{\alpha-D}}\right) \Bigr],
		\end{align}
	
	for an arbitrarily small $\epsilon>0$.
	
	Another tool that is useful to study Gibbs states whenever a Lieb-Robinson bound holds for $H$ is the \emph{quantum belief propagation} (see \cite{Hastings.2007,Kim_2012, capel2023lppl}). 
	Consider two self-adjoint operators $H$ and $V$, and the path of Hamiltonians $H(s) \coloneqq H + sV $. Then, the exponentials $e^{-\beta H(s)}$ satisfy the differential equation
	\begin{equation}\label{eq:diff_exponential}
		\frac{d}{ds}e^{-\beta H(s)} = -\frac{\beta}{2} \Bigl\lbrace e^{-\beta H(s)} , \Phi_\beta^{H(s)} (V) \Bigr\rbrace \, ,
	\end{equation}  
	with the belief propagation operator
	\begin{equation}\label{eq:Phi}
		\Phi_\beta^{H(s)} (V)\coloneqq \int_{-\infty}^{\infty} dt f_\beta (t) e^{-it H(s)} V e^{it H(s)} \, ,
	\end{equation}     
	where $f_\beta(t)=\frac{2}{\beta \pi}\log(\frac{e^{\pi |t|/\beta}+1}{e^{\pi |t|/\beta}-1})$, such that
	\begin{equation}\label{eq:normPhi}
		\norm{\Phi_\beta^{H(s)} (V)} \leq \norm{V} \, .
	\end{equation}
	
	We identify two complementary regimes of interest for our results, which we consider in order to prove our main result.
		\begin{assumptions}
			We differentiate the following cases:
			\begin{enumerate}
				\item \textbf{High-temperature regime:} with $\beta < \beta^*= \frac{1}{8 u g k}$ and $\alpha > D$, so that the Lieb-Robinson bound of the form of Eq.~\eqref{eq:def-lieb-robinson-bound-prelim} holds. 
				\item \textbf{Strong Lieb-Robinson Bound regime:} $\alpha > 2D$, so that one of the stronger Lieb-Robinson bounds, given  by either Eq.~\eqref{eq:def-lieb-robinson-bound-prelim-2d+1} or Eq.~\eqref{eq:def-lieb-robinson-bound-prelim-2d}, holds.
			\end{enumerate}
			\label{def:regimes}
		\end{assumptions}
	
	That is, for the weaker condition $\alpha >  D$ we are bound to considering high temperatures. In both cases, we assume that decay of correlations holds throughout (although note that high temperature already implies decay of correlations \cite{kim2024arealawslongrange}) The motivation for differentiating the two cases is that we expect the  decay of correlations to hold in relevant regimes beyond high temperatures e.g. for 1D systems \cite{kimura2024decaycorrelations1D}.
	
	\section{Stability to global errors}
	\label{sec:stability}
	Our main focus is on the stability of expectation values to errors in the Hamiltonian of the form
	\begin{equation}
		H \rightarrow H+ \varepsilon V,
	\end{equation}
	where $V$ is a Hamiltonian satisfying the same general assumptions as $H$. We show that local expectation values do not depend on the size of the system on which the perturbation $V$ acts (which can be the entire lattice $\Lambda$) and are controlled only by local quantities and the parameter $\varepsilon$, which is assumed to be small.
	
	\subsection{Analytical results}\label{sec:analytic}Consider a finite lattice $\Lambda$ and long-range Hamiltonians $H$ and $V= \sum_{Z \subset \Lambda} V_Z$ on it satisfying (\ref{eq:k-form-hamiltonian}, \ref{eq:def-long-range}) with power-law decay $\alpha> D$. 
	
	The main aim of this section is to prove that the assumptions in As.~\ref{def:regimes} together with decay of correlations imply stability of long-range Hamiltonians. The main result is formally stated as follows. 
	\begin{theorem}\label{thm:stability}
		Consider $H$ and $V$ as above, and either condition in As.~\ref{def:regimes} holds. Fix any $\varepsilon\geq0$ and assume that the Gibbs states of $H + \varepsilon s   V_\Gamma$ satisfy decay of correlations (\ref{def:decay-correlation}) at inverse temperature $\beta<\infty$ uniformly in $s \in [0,1]$ for any $V_\Gamma= \underset{Z\subset \Gamma}{\sum}  V_Z$ with $\Gamma \subset \Lambda$. Then,
		\begin{align*}
			\Bigl|\tr\bigl[O_A \rho_{\beta}[H]\bigr] - &\tr\bigl[O_A \rho_{\beta}[H+\varepsilon V]\bigr]\Bigr| \\
			&\leq  \varepsilon \frac{vu}{2}  \kappa(\beta)  \|O_A\| k|A|e^{|A|/k+1} ,
		\end{align*}
		for all local observables $O_A \in \mathcal{B}(\mathcal{H}_A)$ and a constant $\kappa(\beta) >0$ depending polynomially on $\beta$ and super-exponentially on $\alpha$. 
	\end{theorem}
	
	Therefore, Theorem \ref{thm:stability} shows that in both regimes of As. \ref{def:regimes} 
	
	\begin{equation}\label{eq::long-range-stability}
		\Bigl|\hspace{-0.03cm}\tr\bigl[\hspace{-0.02cm}O_A \rho_{\beta}[H]\bigr] \hspace{-0.09cm}- \hspace{-0.05cm} \tr\bigl[\hspace{-0.02cm}O_A \rho_{\beta}[H\hspace{-0.05cm}+\hspace{-0.05cm}\varepsilon V]\bigr] \hspace{-0.05cm}\Bigr|  \hspace{-0.09cm}\leq \hspace{-0.08cm}\cO  \hspace{-0.09cm}\Big(  \hspace{-0.06cm}\varepsilon \|O_A\|e^{\frac{|A|}{k}}  \hspace{-0.06cm}\Big)
	\end{equation}
	for all $\varepsilon\geq0$ and local observables $O_A$. The proof of Thm.~\ref{thm:stability} for both cases of As.~\ref{def:regimes} is based on two steps: 
	\begin{itemize}
		\item Use of decay of correlations and the Lieb-Robinson bound to prove the intermediate result named \emph{Local perturbations perturb locally} (LPPL) (cf.~App.~\ref{app:corrtoLPPL}).
		\item Show that LPPL implies that a global error affects local expectations only locally.
	\end{itemize}

	These two steps are also required in the analogous proof for the case of finite-range interactions. However, the expected decay rate of correlations and the prefactors appearing in that case are smaller, which substantially simplifies both steps (cf.~\cite[Theorem 34]{capel2023lppl}). For long-range interactions, these aspects present more subtleties.
	
	For the first step of the proof, consider $\rho_\beta[H]$ the Gibbs state of $H$ in $\Lambda$ at inverse temperature $\beta < \infty$, and a local perturbation $V_B$ supported in $B$. Denote by  $\rho_\beta[H+V_B]$ the Gibbs state of the perturbed Hamiltonian. Then, if we consider an observable  $O_A $ supported in $A$, LPPL (see Def.~\ref{def:stability-local-perturbation} below) states that the expectations of $O_A$ with respect to $\rho_\beta[H]$ and $\rho_\beta[H+V_B]$ are practically indistinguishable when $A$ and $B$ are far apart. In other words, the expectation values of local observables supported far from the perturbation are not influenced by its presence.
	\begin{definition}[LPPL]\label{def:stability-local-perturbation}
		Given a finite lattice $\Lambda$ and an inverse temperature $\beta \geq 0$, we say that $H$ satisfies \textit{LPPL} if there is a constant $K'\geq0$ such that 
		\begin{align*}
			& |\tr[O_A \rho_\beta[H]] - \tr[O_A \rho_\beta[H+V_B]]| \\[1mm]
			&\hspace{2.5cm}\leq K' \|O_A\| \|V_B\| \frac{|A| |B| e^{(|A|+|B|)/k}}{(1+d(A, B))^{\alpha}}\,,
		\end{align*}
		for every $A, B \subset \Lambda$ (see \cref{fig1}), any local perturbation $V_B = V_B^\dagger \in \cB(\cH_B)$ and any $O_A \in \cB(\cH_A)$.
	\end{definition}
	The notion of LPPL concerns local perturbations of quantum systems and is often used to study the stability of Gibbs states \cite{KastoryanoEisert_2014,capel2023lppl} as well as ground states \cite{Bachmann_2011,De_Roeck_2015,Henheik2022}. Using this concept, we present our main technical lemma of independent interest: that LPPL holds in both regimes of As.~\ref{def:regimes} given decay of correlations. 
	\begin{lemma}[cf.~Lem.~\ref{lem:decay-cor-lppl}]\label{lem-main:decay-cor-lppl}
		For $A, B \subset \Lambda$ and $H$, $V_B$ as above,  if $H+sV_B$ satisfy either assumption in As.~\ref{def:regimes} and decay of correlations (\ref{def:decay-correlation}) uniformly in all $s\in[0,1]$, then, 
		\begin{align*}
			\big|\tr[O_A \rho_{\beta}[H]] &- \tr[O_A \rho_{\beta}[H+V_B]] \, \big| \\ & \leq \kappa(\beta) \|O_A\| \|V_B\| \frac{|A| |B| e^{(|A|+|B|)/k}}{(1+d(A, B))^{\alpha}}\,,
		\end{align*}
		where $O_A\in\cB(\cH_A)$ and $\kappa(\beta)\geq 0$ is polynomial in $\beta$ and in the constants of the Lieb-Robinson bound.
	\end{lemma}
	
		The proof is given in App.~\ref{app:corrtoLPPL}, and it distinguishes between the two regimes of As.~\ref{def:regimes}, as we now summarize. We separate the argument into short- and long-distance contributions. For short distances, the contribution to the matrix exponential of $V_B$ is bounded directly using the perturbation formula in Eq.~\eqref{eq:diff_exponential}. For long distances, the same perturbation can be written in terms of a generalized covariance~\cite{Kliesch.2014} via Duhamel's formula, which we then relate to the correlation function in Eq.~\eqref{eq:def-correlation-function} using the Lieb--Robinson bound. The result follows by assuming decay of correlations to estimate the resulting integrals.
	
	At this point, a difference in the proof arises between the two cases of As.~\ref{def:regimes}. For $D < \alpha \leq 2D$, since we use a weaker Lieb-Robinson bound (Eq.~\eqref{eq:def-lieb-robinson-bound-prelim}), we need to assume a high-temperature regime $\beta < \beta^*$. In contrast, for $\alpha > 2D$ the stronger Lieb--Robinson bounds (Eqs.~\eqref{eq:def-lieb-robinson-bound-prelim-2d+1} and \eqref{eq:def-lieb-robinson-bound-prelim-2d}) allow us to bound these contributions independently of the temperature. Note, however, that for $2D < \alpha \leq 2D+1$ we obtain a weaker decay $\propto 1/r^{\alpha - D}$.
	
		Thus, if $\alpha > 2D$, LPPL can be established whenever decay of correlations holds, independently of the temperature. This includes, in particular, the case of $1D$ systems at arbitrary temperatures~\cite{capel2023lppl}. For $D>1$, however, this still relies on the validity of decay of correlations in a regime that may lie beyond currently known proofs~\cite{kim2024arealawslongrange}.

	As the second step toward proving Thm.~\ref{thm:stability}, we show that local expectation values are stable against perturbations by sums of small local terms in the underlying Hamiltonian, assuming LPPL. This result was previously proven for finite and short-range interactions in \cite{Rakovszky2024,capel2023lppl,Trivedi2024}.  
	\begin{figure}[t!]
		\begin{center}
			\includegraphics[scale=0.17]{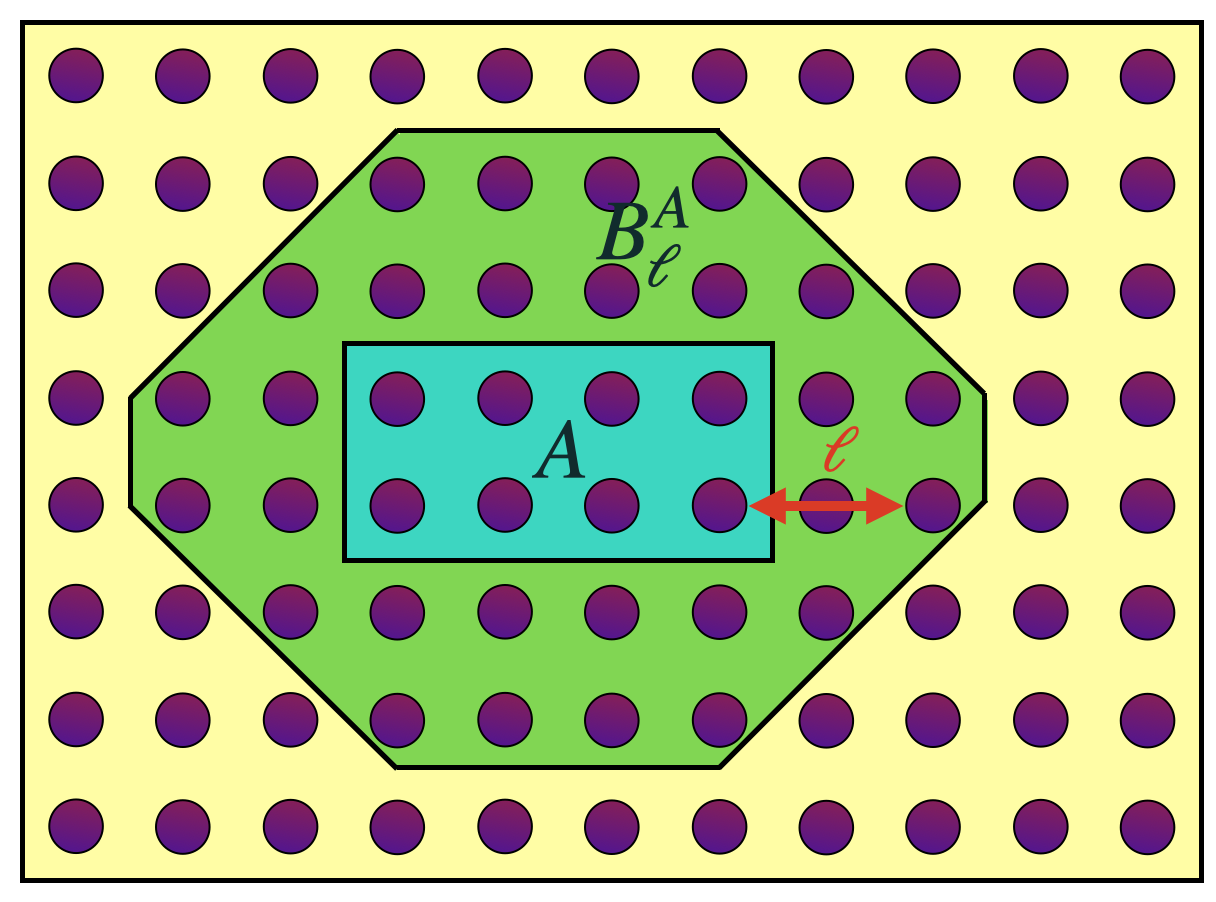}  \hspace{0.1cm}\includegraphics[scale=0.17]{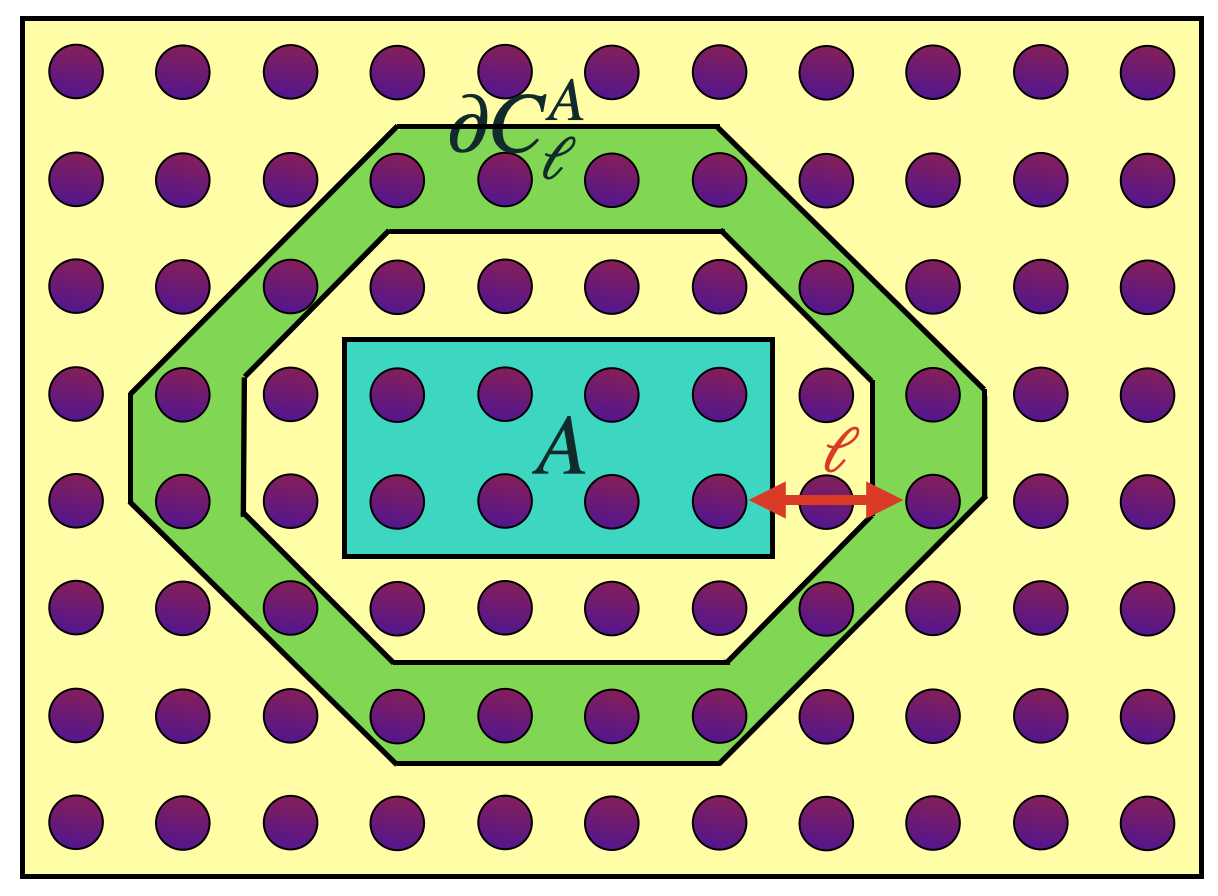}
			\caption{\justifying For a region $A \subset \Lambda$, graphical definition of $B_\ell^A$ and $\partial C^A_\ell$ for $\ell=2$.}
			\label{fig2}
		\end{center}
	\end{figure}  
	\begin{lemma}[cf.~Lem.~\ref{lem:lppl-local-indist}]\label{lem:SLT}
		Assume that $H + \varepsilon  V_\Gamma$ defined above, for any $V_\Gamma=\sum_{Z\subset \Gamma}V_Z$ with $\Gamma \subset \Lambda$, satisfies the LPPL (\ref{def:stability-local-perturbation}) for $\alpha>D$.  For a given $\varepsilon\geq0$, we have
		\begin{align*}
			|\tr[O_A \rho_{\beta}[H]] - &\tr[O_A \rho_{\beta}[H+\varepsilon V]]|\\
			& \leq\varepsilon g u^2 k K'  \|O_A\|\,|A|e^{|A|/k+1},
		\end{align*}
		where $O_A\in\cB(\cH)$ and $K' \geq 0$, which is polynomial in the Lieb Robinson bound constants, from Def.~\ref{def:stability-local-perturbation}. 
	\end{lemma}
	\begin{proof}
		We start by defining the level sets (see \cref{fig2})
		\begin{equation*}
			\partial C_\ell^A=\{x\in A^c\,|\,d(A,x)=\ell\}\,
		\end{equation*}
		for all $\ell\in\N$ and 
		\begin{equation*}
			B_\ell^A=\{x\in A^c\,|\,d(A,x)\leq\ell\}\,.
		\end{equation*}
		Next, we delete the perturbation terms $V_Z$ level by level: First, we fix a level set by $\ell\in\N$ and consider all perturbations intersecting $\partial C_\ell^A$ and with an element $B_\ell^A$ of minimal distance to $A$. By using a telescopic sum, the level sets and the LPPL assumption, we show 
		\begin{align*}
			&|\tr[O_A \rho_\beta[H]] - \tr[O_A \rho_\beta[H+\varepsilon V]]\,|\\
			&\leq\sum_{\ell=0}^\infty\sum_{\substack{Z\subset \Lambda\,|\,Z\ni\{b,c\}\in B_\ell^A\times \partial C_{\ell}^A \\ \qquad \text{s.t.~}d(A,Z)=d(A,b)}} \hspace{-0.5cm} \!\!\!K' \varepsilon  \|V_Z\| \|O_A\|  \frac{k|A|e^{|A|/k+1}}{(1+d(A,b))^{\alpha}}\\
			&\leq k K' \varepsilon g \|O_A\|\,|A|e^{|A|/k+1}\\ 
			& \quad \quad \times \sum_{\ell=0}^\infty\sum_{c\in\partial C_{\ell}^A}\sum_{b\in B_\ell^A}\frac{1}{(1+d(b,c))^\alpha(1+d(A,b))^{\alpha}}\,.
		\end{align*}
		Then, we apply the uniform summability and convolution properties of the polynomial decay for $\alpha>D$ \cite{Nachtergaele_2006,hastings2010quasiadiabatic} (cf.~Lem.~\ref{lem:convolution}), which show
		\begin{align*}
			&|\tr[O_A \rho_\beta[H]] - \tr[O_A \rho_\beta[H+\varepsilon V]]| \\
			&\leq \varepsilon g K'  \|O_A\| k|A|e^{|A|/k+1}\sum_{\ell=0}^\infty\sum_{c\in\partial C_{\ell}^A}\frac{ u}{(1+d(A,c))^\alpha}\\
			&\leq \varepsilon g u^2 K'  \|O_A\| k|A|e^{|A|/k+1} .
		\end{align*}
	\end{proof}
	\begin{proof}[Proof of Thm.~\ref{thm:stability}]
		Since $V = \sum_{Z \subset\Lambda} V_Z$, for any of these $V_Z$,  by Lem.~\ref{lem-main:decay-cor-lppl}, assuming decay of correlations for the Gibbs states of $H+s\varepsilon V_Z$, and our assumptions in As.~\ref{def:regimes}, we have
		\begin{align*}
			\big|\tr[O_A \rho_{\beta}[H]] - &\tr[O_A \rho_{\beta}[H+\varepsilon V_Z]] \, \big| \\ & \leq \varepsilon \kappa(\beta)  \|O_A\| \|V_Z\| \frac{|A| k e^{(|A|+k)/k}}{(1+d(A, B_Z))^{\alpha}} \, ,
		\end{align*}
		where $B_Z$ is the support of $V_Z$, with cardinality $|B_Z|=k$. We conclude by using Lem.~\ref{lem:SLT}.
	\end{proof}     
	
	\subsection{Numerical results}
	\label{sec:numerical}
	Our analytical results hold in the long-range regime $\alpha > D$, and under additional assumptions of Lieb-Robinson bounds. However, multiple relevant quantum platforms have interactions with stronger long-range interactions, often even satisfying $\alpha \leq D$ \cite{reviewLR_RevModPhys.95.035002}. Since an analytical approach is still missing for this regime, this motivates employing numerical methods in its exploration. The purpose is twofold: to investigate a regime not covered by our analytical treatment and to further test the accuracy of our analytical bounds within their domain of applicability.
	
	We perform numerical simulations of the following system defined on a spin chain $(D = 1)$:
	\begin{equation}\label{eq:hamiltonian}
		H= H_{\text{LR,Is}}+ \varepsilon V \,,
	\end{equation}
	with the long-range transverse-field Ising model
	\begin{equation}
		H_{\text{LR,Is}} = \frac{J}{\mathcal{N}_{LR}} \sum_{i<j} \frac{\sigma_i^x \sigma_j ^x}{\vert i-j\vert^\alpha} +  h \sum_i  \sigma_i^z, 
		\label{eq::HLRising}
	\end{equation}
	with  $\mathcal{N}_{LR}= N^{-1}\sum_{i\leq j} \vert i-j\vert^{-\alpha}$, and the perturbation
	\begin{equation}
		V=  \sum_i  \sigma_i^z \, . 
	\end{equation}
	From an experimental perspective, this perturbation can simply correspond to a miscalibration of the magnetic field in the $z$ direction.
	\begin{figure*}
		\centering
		\includegraphics[width=0.95\linewidth]{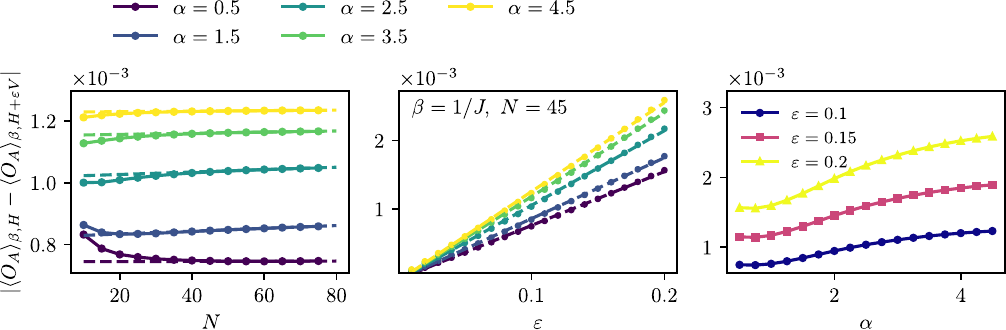}

		\caption{\justifying Absolute value of the difference in the expectation value of a local observable $O_A=\sigma_i^z \sigma_{i+1}^z$ at the middle of the system between perturbed and unperturbed Gibbs states of Hamiltonian in Eq.~\eqref{eq:hamiltonian}. This difference is plotted as a function of system size fixing $\varepsilon=0.1$ (left), and then with fixed system size N=45, against the magnitude of the perturbation $\varepsilon$ (middle) and against $\alpha$ (right). All plots are at fixed inverse temperature $1/J$. We show both numerical simulations (points) and linear fits (dashed lines). See Appendix \ref{app:sim} for the specific values and parameters of the simulations and further discussion of the results.}
		\label{fig:DifvsEverything}
	\end{figure*}
	
	The results are shown in Fig.~\ref{fig:DifvsEverything}, where
	we observe that the effects of global perturbations, when computing local expectations, tend to stabilize as the system size increases, as expected from the analytical results above. It is worth highlighting that this behavior persists even in the strong long-range regime $\alpha \leq D=1$, which goes beyond all our assumptions.
	
	Additionally, Fig.~\ref{fig:DifvsEverything} shows that the difference between expectation values grows linearly with $\varepsilon$, in agreement with the upper bound provided by Thm.~\ref{thm:stability}. This indicates that the bound is optimal with respect to the perturbation magnitude.
	We also show the effect of the perturbation as a function of $\alpha$, and we see that the more local the system (i.e., the larger the value of $\alpha$), the more pronounced the impact of perturbations becomes. For further discussion on the simulations performed as well as additional plots, we refer the reader to App.~\ref{app:num}.
	
	\section{Local indistinguishability and clustering}
	In order to obtain the results in Sec.~\ref{sec:analytic} it is crucial to consider the notions of decay of correlations and LPPL. We complete the picture of clustering in long-range systems by considering the important and closely related notion of \textit{local indistinguishability} \cite{Brandao2019}. 
	For simplicity, we here assume that the lattice is hypercubic, but the results hold for any other $D-$dimensional lattice simply by changing constant factors.
	This property provides a decay bound on the error made when local expectations of a global Gibbs state are approximated by a localized one. It is also known as \emph{locality of temperature} \cite{Kliesch.2014,Hern_ndez_Santana_2015}.
	\begin{figure}[t!]
		\begin{center}
			\includegraphics[scale=0.25]{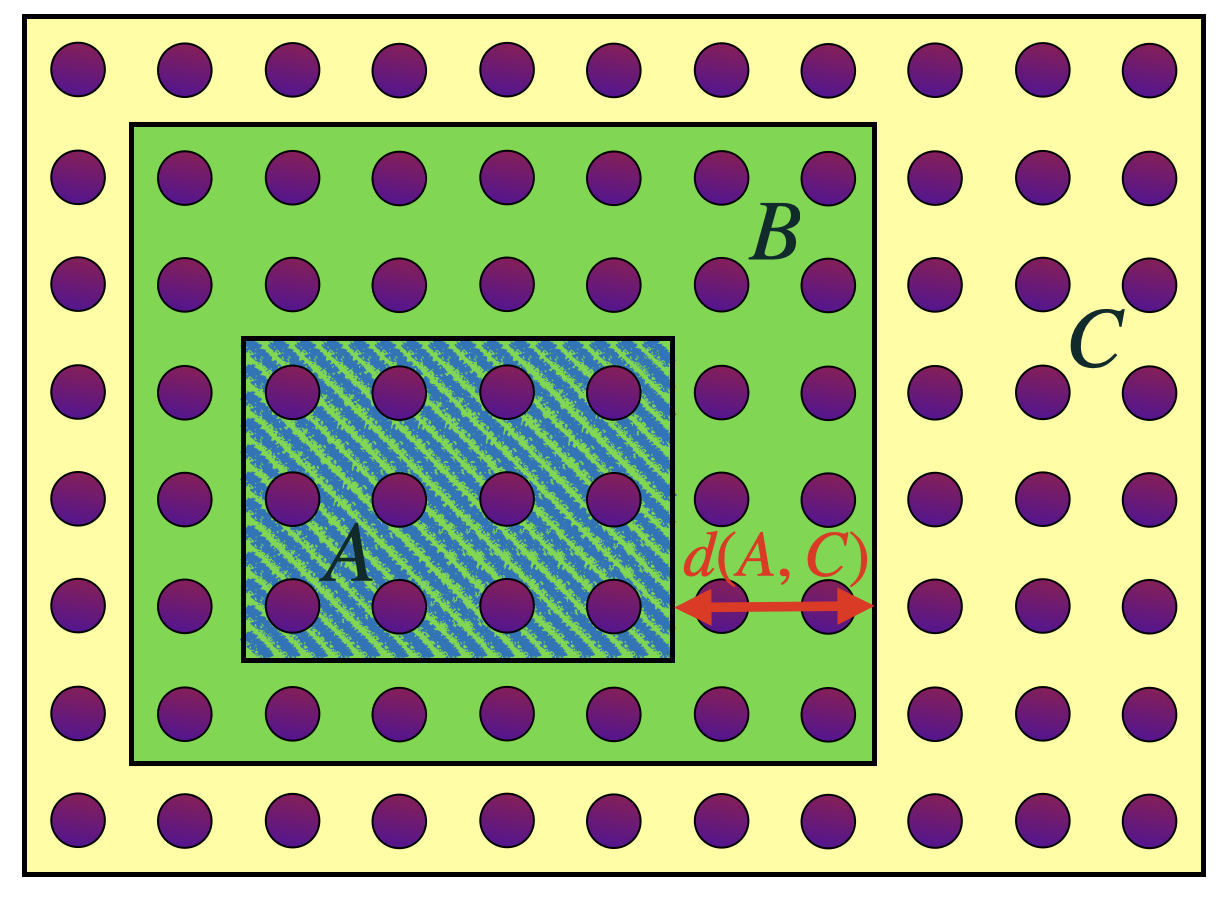}
			\caption{A square lattice $\Lambda$ split into $B$ and $C$. The region $A$, signaled with diagonal blue lines, is a subregion of $B$.}
			\label{fig4}
		\end{center}
	\end{figure}  
	\begin{definition}[Local indistinguishability]\label{def:local-indistinguishability}
		Consider $A \subset B \subset \Lambda$ as in Fig.~\ref{fig4} and $H$ as above. Then, the Gibbs state is locally indistinguishable if
		\begin{align*}
			\big|\tr[O_A \rho_\beta[H]] & - \tr[O_A \rho_\beta[H_B]] \, \big| \\ &\leq K'' \|O_A\| \frac{|A| e^{|A|/k}}{(1+d(A, B^c))^{\alpha}},
		\end{align*}
		for every $O_A \in \cB(\cH_A)$, with a constant $K''\geq0$, where $H_B \coloneqq \sum_{Z \subset B } h_Z$.
	\end{definition}
	Thus, expectation values of local observables can be computed using the Gibbs state on a smaller vicinity \cite{Kliesch.2014,Alhambra_2021}. It also has relevant implications for the preparation of Gibbs states \cite{Brandao2019}. We prove that this property also follows from LPPL, albeit with a slightly worse dependence on $\vert A \vert$ and $d(A, C)$ than that of Def.~\ref{def:local-indistinguishability}.
	
	\begin{lemma}(cf.~Lem.~\ref{lem:lppl-local-indist})\label{lem-main:lppl-local-indist}
		Assume that a Hamiltonian $H$ as in Eq.~(\ref{eq:k-form-hamiltonian}, \ref{eq:def-long-range}) for $\alpha> D$ satisfies the LPPL property (\ref{def:stability-local-perturbation}) for all $V\in\{H_E\,|\,E\subset\Lambda\}$. Then, for all $\delta\in(0,\alpha-D)$
		\begin{align*}
			|\tr[O_A \rho_\beta[H]]& - \tr[O_A \rho_\beta[H_B]]| \\[1mm] &\leq \kappa' \zeta(1+\delta) \|O_A\| \frac{|A|^2 e^{|A|/k}}{(1+d(A, C))^{\alpha-D-\delta}},
		\end{align*}
		where $A \subset B \subset \Lambda$, $H_B = \sum_{Z \subset B } h_Z$, $C \coloneqq \Lambda \setminus B $, $\kappa'\geq 0$ and $\zeta$ the Riemann zeta function, e.g.~$\zeta(2)=\frac{\pi^2}{6}$. 
	\end{lemma}
	The proof is shown in App.~\ref{app:lppltoloc}. It involves an iterative elimination of the Hamiltonian terms coupling $B$ with its exterior, which can be done through the LPPL property. Then, the total sum of terms can again be controlled by the convolution and uniform summability properties of the long-range decay. The resulting sum only decays with a weakened exponent $\alpha-D$. As in Sec.~\ref{sec:analytic}, these implications had already been proven for finite and short-range interactions \cite{Brandao2019,capel2023lppl,Rouze2024Learning}, but the proof in the long-range case is more subtle.
	
	To complete the picture, we also explore local indistinguishability numerically in Fig.~\ref{fig:locindist} for the model in Eq.~\eqref{eq::HLRising}. We again find that it holds for ranges of $\alpha$ beyond our present results, and into the strong long-range regime $\alpha \le D$. The decay exponents also appear to be larger than those predicted in the bounds and Def.~\ref{def:local-indistinguishability}. 
	\begin{figure}[t!]
		\centering
		\includegraphics[width=0.85\linewidth]{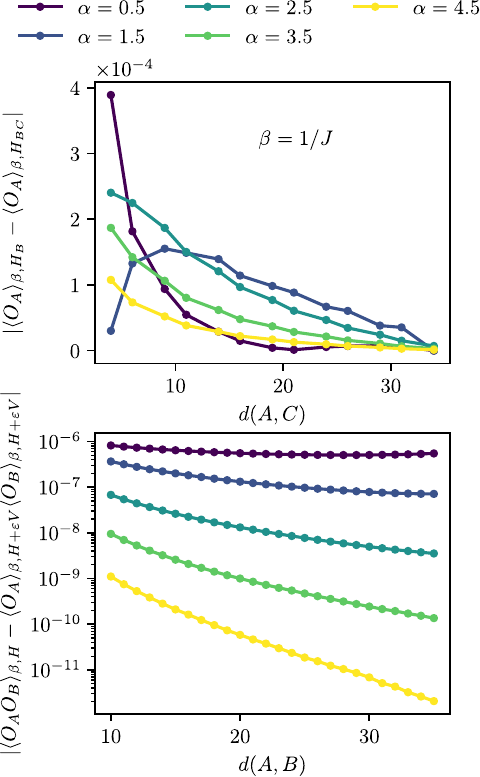}
		\caption{\justifying Test of local indistinguishability (top) and decay of correlations (bottom) in the model of Eq.~\eqref{eq::HLRising}.  For local indistinguishability, we show the difference between thermal mean values of the observable $O_A$ in the middle of the chain of $N=75$ sites in regions $B$ and $BC$, as illustrated in Fig.~\ref{fig4}, against the distance between the frontiers of both regions. We again see that it holds for a wider range of $\alpha$ than Lem.~\ref{lem-main:lppl-local-indist}. For the decay of correlations, we locate the observable $O_A=\sigma_i^z \sigma_{i+1}^z$ at $i=4$ and move $O_B$ to the right of the chain to test the dependence with distance. }
		\label{fig:locindist}
	\end{figure}

	An interesting feature of local indistinguishability is that it allows us to close the circle of implications with the starting point being decay of correlations, although also with a slower decay with the distance $d(A,B)$ than the original Def.~\ref{def:decay-correlation}. 
	\begin{lemma}(cf.~Lem.~\ref{lem:local-indist-decay-cor})\label{lem-main:local-indist-decay-cor}
		Assume that a Hamiltonian $H$ as in Eq.~(\ref{eq:k-form-hamiltonian}, \ref{eq:def-long-range}) satisfies the local indistinguishability (\ref{def:local-indistinguishability}) with $\alpha>2D$. Then,
		\begin{equation*}
			\begin{aligned}
				\mathrm{Cov}_{\rho_\beta}(O_A, O_B)
				&\leq \kappa''\|O_A\|\,\|O_B\|\, \\ \times &|A|\,|B|e^{(|A|+|B|)/k}\frac{1}{(1+d(A,B))^{\alpha-2D}}
			\end{aligned}
		\end{equation*}
		where $O_A, O_B\in\cB(\cH)$ and $\kappa''\geq 0$ a constant. 
	\end{lemma}
	The proof is shown in App.~\ref{app:loctodec}. It follows the proof for finite-range from \cite{capel2023lppl} by applying local indistinguishability to the joint observable $O_A O_B$, and to each of $O_A,O_B$ separately, done by considering buffer regions $X,Y$. To compare the resulting expectation values, we need to bound the contribution of terms that couple the buffer regions $X,Y$, which we have to pick large enough so that local indistinguishability holds, but small enough so that the terms that couple them are small. The optimal choice leads to the decay with exponent $\alpha-2D$, slower than that those in Def.~\ref{def:decay-correlation} and of Fig.~\ref{fig:locindist} (see also Fig.~\ref{fig:appendix} in the Appendix). This equivalence is expected to be physically meaningful whenever correlations decay with an effective exponent $\alpha_{\mathrm{cor}} > 2D$, as observed in certain long-range models~\cite{kitaevlrchain_PhysRevLett.119.110601}. In practice, this condition is met for microscopic exponents $\alpha > 2D$, as realized in experimental platforms such as Rydberg-atom arrays~\cite{reviewLR_RevModPhys.95.035002}.

	This last result closes a circle of implications between the different properties. In finite-range systems, where correlations often decay exponentially, it is known that the exponential behavior carries throughout the circle of implications \cite{capel2023lppl}. Here, instead, the long-range nature of the interactions prevents us from such a clean equivalence of decays, since the polynomial overheads we incur in with the proof affect the final scaling of the bound, unlike in finite range where decaying exponentials always dominate. See Fig.~\ref{fig3} for an illustration. 
	
	\vspace{0.5cm}
	
	\begin{figure}[t!]
		\vspace*{3ex}
		\begin{center}
			\includegraphics[scale=0.18]{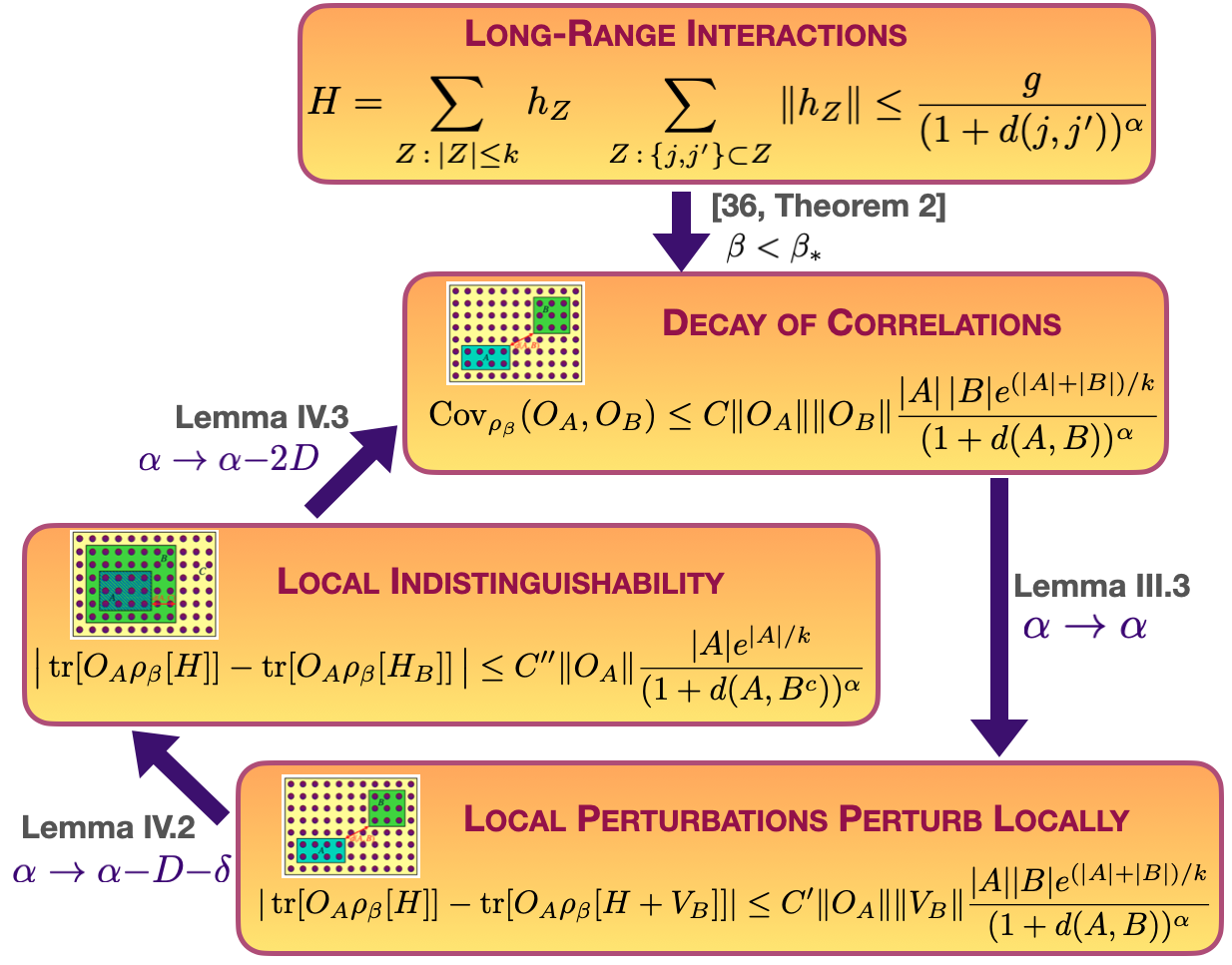}
			\caption{Equivalence between decay of correlations, LPPL and local indistinguishability for Hamiltonians with long-range interactions. Note that $\alpha$ changes in some of the implications, as indicated next to the corresponding arrows. As a result, completing a full cycle --- starting and ending, for example, at decay of correlations ---yields a different version of the property. At high enough temperature, decay of correlations, and subsequently all properties, hold.}
			\label{fig3}
		\end{center}
	\end{figure}
	
	\section{Conclusion}
	We have studied the relation between the stability of expectation values to global errors and various properties related to clustering of correlations, with the extra technical difficulty of assuming that the interactions are long-range. This difficulty required a different proof method to that of finite-range systems, and introduces several overheads. The most natural next step is to find arguments that at the very least reach the regimes of stability and clustering that we see in the numerical examples studied. Other directions in which to extend the results are to long-range bosonic \cite{BosonsLemm2023}, fermionic \cite{GongGuaita2023,Trivedi2024} or open quantum systems \cite{Roon_2024,kashyap2024quantumadvantage}.
	
	With our results we give a piece of evidence of the robustness to coherent errors of analog simulation experiments. They thus motivate the quantum simulation of Gibbs states of long-range models. A potential further motivation is that, even at high temperatures, it is unclear whether it is possible to efficiently calculate local expectation values and partition functions \cite{segovia2025hightemperature}, as opposed to what happens in finite-range interactions \cite{Mann2021,Harrow_2020,Haah2024}.
	
	We expect that our results could also be of use in the context of learning Gibbs states of long-range Hamiltonians within phases of matter as in \cite{Rouze2024Learning}. There, relations between properties akin to those here, but for finite-range interactions, were crucial in the proof.
	\begin{acknowledgments}
		TM and AC acknowledge the support of the Deutsche Forschungsgemeinschaft (DFG, German Research Foundation) - Project-ID 470903074 - TRR 352 and funding by the Federal Ministry of Education and Research (BMBF) and the Baden-Württemberg Ministry of Science as part of the Excellence Strategy of the German Federal and State Governments. The work of JSS was supported by the Spanish Research Council, CSIC, through the JAE-PREDOC 2024 and JAE-INTRO 2023 programs. AMA acknowledges support from the Spanish Agencia Estatal de Investigaci\'on through the grants ``IFT Centro de Excelencia Severo Ochoa CEX2020-001007-S", ``PCI2024-153448" and ``Ram\'on y Cajal RyC2021-031610-I", financed by MCIN/AEI/10.13039/501100011033 and the European Union NextGenerationEU/PRTR. This project was funded within the QuantERA II Programme that has received funding from the EU’s H2020 research and innovation programme under the GA No 101017733. We also acknowledge support from the Open Access Publishing Fund of the University of T\"{u}bingen.
	\end{acknowledgments}
	
	\sloppy
	\bibliography{references}
	\newpage
	\appendix
	\widetext
	\section{Technical proofs}
	We now show the proofs of the relations between decay of correlations (Definition \ref{def:decay-correlation}), stability against local perturbations (Definition \ref{def:stability-local-perturbation}), and local indistinguishability (Definition \ref{def:local-indistinguishability}).

	\subsection{From decay of correlations to LPPL}\label{app:corrtoLPPL}
	To cover the cases in \Cref{def:regimes}, we first establish a slightly more abstract result for any Lieb–Robinson bound that admits the following structure for $t\leq\frac{d(A,B)}{2v}\eqqcolon b$ with $v=2gu$ given disjoint regions $A,B\subset\Z^D$:
		\begin{equation}\label{eq:lieb-robinson-bound-appx}
			\|[O_A(t),O_B]\|\leq \kappa_{LR}\|O_A\|\|O_B\|\frac{g(t)}{(1+d(A,B))^{\alpha}}\,.
		\end{equation}
		for any $O_A \in \cB(\cH_A) $ and $O_B \in \cB(\cH_B) $, $\kappa_{LR}$ a constant depending on $\alpha, g$ and $D$, and with $O_A(t)\coloneqq e^{itH}O_Ae^{-itH}$. Moreover, assume that $g:[0,\infty)\rightarrow[0,\infty)$ is a integrable function satisfying
		\begin{equation*}
			\int_{0}^{b} \frac{e^{-2\pi t/\beta}g(t)}{1 - e^{-2\pi t /\beta}} dt\leq I_c(\beta)
		\end{equation*}
		given a polynomial $I_c(\beta)$ independent  of $b$. In the different cases of Assumption~\ref{def:regimes}, \Cref{cor:decay-cor-lppl-temp}-\ref{cor:decay-cor-lppl-2D-2D+1} directly follow proving the result ``From decay of correlations to LPPL'' for the specific cases.
		\begin{lemma}\label{lem:decay-cor-lppl}
			Let $\Lambda$ be a finite lattice and $A,B\subset \Lambda$. Consider the Hamiltonian $H$ defined as in \eqref{eq:k-form-hamiltonian} and \eqref{eq:def-long-range}, and for any $V_B=V_B^\dagger$ supported locally on $B\subset \Lambda$, define $H+sV_B$ for any  $s\in[0,1]$. Assume that for a given $\alpha >0$, $H+sV_B$ satisfies decay of correlations (\ref{def:decay-correlation}) and a Lieb-Robinson bound of the form \eqref{eq:lieb-robinson-bound-appx} uniformly in all $s\in[0,1]$. Then,
			\begin{equation*}
				\big|\tr[O_A \rho_{\beta}[H]] - \tr[O_A \rho_{\beta}[H+V_B]] \, \big| \leq \kappa \left(\beta\right) \|O_A\| \|V_B\| \frac{|A| |B| e^{(|A|+|B|)/k}}{(1+d(A, B))^{\alpha}}\,,
			\end{equation*}
			where $O_A\in\cB(\cH)$ and $\kappa \left(\beta\right)\geq 0$ a constant depending in particular polynomial on $\beta$ but also on the constants of the Lieb-Robinson bound.
	\end{lemma}
	We start by listing the corollaries that follow, before presenting a detailed proof of \Cref{lem:decay-cor-lppl}.
	\begin{corollary}\label{cor:decay-cor-lppl-temp}
			Let $\Lambda$ be a finite lattice and $A,B\subset \Lambda$. Consider the Hamiltonian $H$ defined as in \eqref{eq:k-form-hamiltonian} and \eqref{eq:def-long-range}, and for any $V_B=V_B^\dagger$ supported locally on $B\subset \Lambda$, define $H+sV_B$ for any  $s\in[0,1]$. Assume that $\beta < \beta^*= \frac{1}{8 u g k}$ and $\alpha >D$. Then,
			\begin{equation*}
				\big|\tr[O_A \rho_{\beta}[H]] - \tr[O_A \rho_{\beta}[H+V_B]] \, \big| \leq \kappa \left(\beta\right) \|O_A\| \|V_B\| \frac{|A| |B| e^{(|A|+|B|)/k}}{(1+d(A, B))^{\alpha}}\,,
			\end{equation*}
			where $O_A\in\cB(\cH)$ and $\kappa \left(\beta\right)\geq 0$ a constant depending in particular polynomial on $\beta$.
	\end{corollary}
	\begin{proof}
			Due to the assumption, $\beta < \beta^*= \frac{1}{8 u g k}$ and $\alpha>D$, the work \cite{hastings2006gapdecaycorrelations} proved the Lieb-Robinson bound
			\begin{equation}\label{eq:def-lieb-robinson-bound-prelim2}
				\|[O_A(t),O_B]\|\leq \kappa_{LR}\|O_A\|\|O_B\|\frac{e^{v|t|}-1}{(1+d(A,B))^{\alpha}},
			\end{equation}
			for any $O_A \in \cB(\cH_A) $ and $O_B \in \cB(\cH_B) $, with $v= 2gu$, and $\kappa_{LR}$ a constant depending on $\alpha, g$ and $D$, regions $A,B \subset \mathbb{Z}^D$ and with $O_A(t)\coloneqq e^{itH}O_Ae^{-itH}$. Moreover, the recent result in \cite[Thm.~2]{kim2024arealawslongrange} proves decay of correlations (see Def.~\ref{def:decay-correlation}). To apply \Cref{lem:decay-cor-lppl}, we need to bound the integral
			\begin{equation*}
				\int_{0}^{b} \frac{e^{-2\pi t/\beta}g(t)}{1 - e^{-2\pi t /\beta}} dt\leq I_c(\beta)
			\end{equation*}
			for $g(t)=e^{vt}-1$. In a first step, we use the inequality $\frac{1}{1 - e^{-x}} \leq 1 + \frac{1}{x}$, which shows
			\begin{align*}
				\int_{0}^{b} \frac{e^{-2\pi t/\beta}g(t)}{1 - e^{-2\pi t /\beta}} dt &\leq \int_0^{b} e^{-2\pi t / \beta} \left( 1+ \frac{\beta}{2 \pi t }\right)  \left(e^{v t}-1\right) dt
				= \int_0^{b} e^{-2\pi t / \beta} \left(e^{v t}-1\right) dt+ \int_0^{b} e^{-2\pi t / \beta} \left(e^{v t}-1\right) \frac{\beta}{2 \pi t }dt \,.
			\end{align*}
			Next we drop the negative term in the first integral and apply $e^{vt}-1 = vt \int_{0}^{1} e^{s v t} d s \leq vt e^{vt}$, showing
			\begin{align}\label{eq:upper-bound-lrb-integral}
				\int_{0}^{b} \frac{e^{-2\pi t/\beta}g(t)}{1 - e^{-2\pi t /\beta}} dt &\leq \left(1+ \frac{v\beta}{2\pi} \right)\!\! \int_0^{b}\!e^{-\left( 2\pi  / \beta- v\right)t} dt = \left(1+ \frac{v\beta}{2\pi} \right) \frac{\beta}{2\pi -v\beta} \left(1-e^{-\left(2\pi/ \beta-v\right)b} \right) \leq \beta \left(1+ \frac{v\beta}{2\pi} \right) ,
			\end{align}
			where in the last inequality we assumed that we are in the high-temperature regime $\beta < \frac{1}{8 u g k}$ so $2\pi- v\beta \ge 1$. Since the above bound is true for any $b>0$, it completes the proof with $b=\frac{d(A,B)}{2v}$ assumed in Eq.~\eqref{eq:lieb-robinson-bound-appx}.
	\end{proof}
	
	\begin{corollary}\label{cor:decay-cor-lppl-2D+1}
			Let $\Lambda$ be a finite lattice and $A,B\subset \Lambda$. Consider the Hamiltonian $H$ defined as in \eqref{eq:k-form-hamiltonian} and \eqref{eq:def-long-range}, and for any $V_B=V_B^\dagger$ supported locally on $B\subset \Lambda$, define $H+sV_B$ for any  $s\in[0,1]$. Assume that $\alpha >2D+1$. Then, 
			\begin{equation*}
				\big|\tr[O_A \rho_{\beta}[H]] - \tr[O_A \rho_{\beta}[H+V_B]] \, \big| \leq \kappa \left(\beta\right) \|O_A\| \|V_B\| \frac{|A| |B| e^{(|A|+|B|)/k}}{(1+d(A, B))^{\alpha}}\,,
			\end{equation*}
			where $O_A\in\cB(\cH)$ and $\kappa \left(\beta\right)\geq 0$ a constant depending in particular polynomial on $\beta$.
	\end{corollary}
	\begin{proof}
		Due to the assumption $\alpha>2D+1$, the paper \cite{Kuwahara.2020lrb} proves that
		\begin{equation*}
			\|[O_A(t),O_B]\|\leq \kappa_{LR}\|O_A\|\|O_B\|\frac{p(|t|)}{(d(A,B)\!-\!v|t|)^{\alpha}},
		\end{equation*}
		for $p(t)\coloneqq|t|^{2D+1}$, any $O_A \in \cB(\cH_A) $ and $O_B \in \cB(\cH_B)$, and $\kappa_{LR}$ a constant depending on $\alpha, g$ and $D$, regions $A,B \subset \mathbb{Z}^D$ and with $O_A(t)\coloneqq e^{itH}O_Ae^{-itH}$. Moreover, the recent result in \cite[Thm.~2]{kim2024arealawslongrange} proves decay of correlations (see Def.~\ref{def:decay-correlation}). To apply \Cref{lem:decay-cor-lppl}, we use $b=\frac{d(A,B)}{2v}$ so that 
			\begin{equation*}
				\|[O_A(t),O_B]\|\leq \kappa_{LR}\|O_A\|\|O_B\|\frac{p(|t|)}{(\frac{1}{2}d(A,B)\!)^{\alpha}}\leq\kappa_{LR}4^{\alpha}\|O_A\|\|O_B\|\frac{p(|t|)}{(1+d(A,B)\!)^{\alpha}},
			\end{equation*}
			where we used $x\geq \frac{1}{2}(x+1)$ for $x\geq1$ in the last inequality. Next, we bound the following integral
			\begin{equation*}
				\int_{0}^{b} \frac{e^{-2\pi t/\beta}g(t)}{1 - e^{-2\pi t /\beta}} dt\leq I_c(\beta)
			\end{equation*}
			for $g(t)=p(t)$. In a first step, we use the inequality $\frac{1}{1 - e^{-x}} \leq 1 + \frac{1}{x}$, which shows
		\begin{equation}\label{eq:int-parts_TI}
			\int_{0}^{b} \frac{e^{-\mu t} t^{\ell}}{1 - e^{-\mu t}} dt \leq \int_{0}^{b} e^{-\mu t} \Bigl(t + \frac{1}{\mu}\Bigr) t^{\ell - 1} dt,
		\end{equation}
		for $\ell = 2D+1$ and $\mu = \frac{2\pi}{\beta}$. For the first integral, by iteratively applying integration by parts, we obtain
		\begin{equation}\label{eq:int-parts_iteration}
			\begin{aligned}
				\int_{0}^{b} e^{-\mu t} t^{\ell} dt &= \frac{-1}{\mu} \Bigl(e^{-\mu b} b^{\ell} - \ell \int_{0}^{b} e^{-\mu t} t^{\ell - 1} dt \Bigr) \\
				&= -e^{-\mu b} \sum_{j=1}^{\ell} \frac{1}{\mu^j} \frac{\ell! b^{\ell - j + 1}}{(\ell - j + 1)!} + \frac{\ell!}{\mu^{\ell}} \int_{0}^{b} e^{-\mu t} dt \\
				&= -e^{-\mu b} \sum_{j=1}^{\ell + 1} \frac{1}{\mu^j} \frac{\ell! b^{\ell - j + 1}}{(\ell - j + 1)!} + \frac{\ell!}{\mu^{\ell + 1}} \\
				&\leq \frac{\ell!}{\mu^{\ell + 1}}
			\end{aligned}
		\end{equation}
		which holds for all $b>0$ and completes the proof.
	\end{proof}
	\begin{corollary}\label{cor:decay-cor-lppl-2D-2D+1}
			Let $\Lambda$ be a finite lattice and $A,B\subset \Lambda$. Consider the Hamiltonian $H$ defined as in \eqref{eq:k-form-hamiltonian} and \eqref{eq:def-long-range}, and for any $V_B=V_B^\dagger$ supported locally on $B\subset \Lambda$, define $H+sV_B$ for any  $s\in[0,1]$. Assume that $2D<\alpha\leq 2D+1$. Then,
			\begin{equation*}
				\big|\tr[O_A \rho_{\beta}[H]] - \tr[O_A \rho_{\beta}[H+V_B]] \, \big| \leq \kappa \left(\beta\right) \|O_A\| \|V_B\| \frac{|A| |B| e^{(|A|+|B|)/k}}{(1+d(A, B))^{\alpha-D}}\,,
			\end{equation*}
			where $O_A\in\cB(\cH)$ and $\kappa \left(\beta\right)\geq 0$ a constant depending in a particular polynomial on $\beta$.
		\end{corollary}
		\begin{proof}
			Due to the assumption $2D < \alpha \leq 2D+1$, the paper \cite{Tran2021} provides the following Lieb-Robinson bound 
			\begin{align*}
				\|[O_A(t),O_B]\|\leq& \|O_A\|\|O_B\|\Bigl[ \kappa_{LR,2} \left(\frac{t}{d(A,B)^{\alpha-2D+\epsilon}}\right)^{\frac{\alpha-D}{\alpha-2D} - \frac{\epsilon}{2}}+ \kappa_{LR,3} \left(\frac{t}{d(A,B)^{\alpha-D}}\right) \Bigr],
			\end{align*}
			for an arbitrarily small $\epsilon>0$. For $\epsilon\leq2$ 
			\begin{equation*}
				\begin{aligned}
					(\alpha-2D+\epsilon)\Bigl(\frac{\alpha-D}{\alpha-2D}-\frac{\epsilon}{2}\Bigr)-(\alpha-D)=\epsilon\Bigl(\frac{\alpha-D}{\alpha-2D}-\frac{\alpha-2D}{2}-\frac{\epsilon}{2}\Bigr)=\frac{\epsilon}{2}\Bigl(\frac{\alpha}{\alpha-2D}-\epsilon\Bigr)\geq\frac{\epsilon}{2}\Bigl(\frac{2D}{1}-\epsilon\Bigr)\geq0
				\end{aligned}
			\end{equation*}
			so that 
			\begin{align*}
				\|[O_A(t),O_B]\
				\leq& 2^{\alpha-D}\max\{\kappa_{LR,2},\kappa_{LR,3}\}\|O_A\|\|O_B\|  \biggl(\frac{t^{\frac{\alpha-D}{\alpha-2D} - \frac{\epsilon}{2}}+t}{(1+d(A,B))^{\alpha-D}}\biggr),
			\end{align*}
			where we used $x\geq \frac{1}{2}(x+1)$ for $x\geq1$ in the last inequality. Next, we choose
			\begin{equation*}
				\epsilon = 2\Bigl(\frac{D}{\alpha-2D}-\Bigl\lfloor\frac{D}{\alpha-2D}\Bigr\rfloor\Bigr)\in(0,2)
			\end{equation*}
			so that with $\frac{\alpha-D}{\alpha-2D}-\frac{\epsilon}{2}=1-\frac{\epsilon}{2}+\frac{D}{\alpha-2D}$
			\begin{align*}
				\|[O_A(t),O_B]\|\leq& 2^{\alpha-D}\max\{\kappa_{LR,2},\kappa_{LR,3}\}\|O_A\|\|O_B\|  \biggl(\frac{t^{1+\lfloor\frac{\alpha-D}{\alpha-2D}\rfloor}+t}{(1+d(A,B))^{\alpha-D}}\biggr).
			\end{align*}
			To apply \Cref{lem:decay-cor-lppl} for the shifted long-range decay $\alpha-D$, we required to bound the following integral
			\begin{equation*}
				\int_{0}^{b} \frac{e^{-2\pi t/\beta}g(t)}{1 - e^{-2\pi t /\beta}} dt\leq I_c(\beta)\qquad\text{for}\qquad g(t)\coloneqq t^{1+\lfloor\frac{D}{\alpha-2D} \rfloor}+t\,.
			\end{equation*}
			Then, Eq.~(\ref{eq:int-parts_TI}-\ref{eq:int-parts_iteration}) show
			\begin{equation*}
				I_c(\beta)=\frac{2!}{\mu^3}+\frac{1!}{\mu}+\frac{\ell!}{\mu^{\ell+1}}+\frac{\ell!}{\mu^\ell}
			\end{equation*}
			for $\ell = \lfloor\frac{D}{\alpha-2D} \rfloor$, $\mu = \frac{2\pi}{\beta}$, and for all $b>0$. This finishes the proof of the statement.
	\end{proof}
	
	\begin{proof}[Proof of \Cref{lem:decay-cor-lppl}]
		The proof splits into two cases:\\[0.1cm]
		\underline{\textbf{Case I: $d(O_A,V_B)=0$}}. 
		We start with the case~$A\cap B\neq \emptyset$. For that we combine the proof idea used in \cite[Thm.~1]{Kliesch.2014} with the quantum belief propagation presented in \cite[Sec.~10]{capel2023lppl}. Denote $H(s)\coloneqq H+s V_{B}$ and $Z_s(\beta)\coloneqq \tr[e^{-\beta H(s)}]$. Using  the fundamental theorem of calculus, \eqref{eq:diff_exponential} for the derivative of $e^{-\beta H(s)}$, and the chain rule in the derivative of $\rho_\beta[H(s)]$,  then 
		\begin{equation*}
			\begin{aligned}
				\Bigl|&\tr[\rho_{\beta}[H]O_A]-\tr \big[\rho_{\beta}[H+V_B]O_A \big]\Bigr|\\
				&\hspace{2cm}=\Bigl|\tr \Big[O_A\int_0^1\frac{d}{ds}\rho_{\beta}[H(s)]ds \Big]\Bigr|\\
				&\hspace{2cm}=\Bigl|\tr \Big[O_A\int_0^1\Bigl(\frac{1}{Z_{s}(\beta)}\frac{d}{ds}e^{-\beta H(s)}-\rho_{\beta}[H(s)]\frac{1}{Z_{s}(\beta)}\tr \Big[\frac{d}{ds}e^{-\beta H(s)}\Big]\Bigr)ds\Big]\Bigr|\\
				&\hspace{2cm}=\frac{\beta}{2}\Bigl|\tr \Big[O_A\int_0^1\Bigl(\bigl\{\rho_{\beta}[H(s)],\Phi^{H(s)}_{\beta}(V_B)\bigr\}-2\rho_{\beta}[H(s)]\tr[\rho_{\beta}[H(s)]\Phi^{H(s)}_{\beta}(V_B)]\Bigr)ds \Big]\Bigr|\\
				&\hspace{2cm}\leq \frac{8}{\pi}\|O_A\|\,\|V_B\|\int_{0}^{\infty}\log(\frac{e^{\pi t/\beta}+1}{e^{\pi t/\beta}-1})dt\, ,
			\end{aligned}
		\end{equation*}
		where the explicit expression for $\Phi^{H(s)}_{\beta}(V_B)$ is given in \eqref{eq:Phi} and the last bound follows from \eqref{eq:normPhi}. Next, we rewrite the integral by substitution with $x=e^{-\pi t/\beta}$
		\begin{equation*}
			\begin{aligned}
				\int_{0}^{\infty}\log(\frac{e^{\pi t/\beta}+1}{e^{\pi t/\beta}-1})dt
				&=-\frac{\beta}{\pi}\int_{0}^{\infty}\log(\frac{1+e^{-\pi t/\beta}}{1-e^{-\pi t/\beta}})\frac{-\frac{\pi}{\beta}e^{-\pi t/\beta}}{e^{-\pi t/\beta}}dt\\
				&=-\frac{\beta}{\pi}\int_{1}^{0}\log(\frac{1+x}{1-x})\frac{1}{x}dx\\
				&=\frac{\beta}{\pi}\int_{0}^{1}\biggl(\frac{\log(1+x)}{x}-\frac{\log(1-x)}{x}\biggr)dx\\
				&=\frac{\beta}{\pi}\biggl(\int_{0}^{-1}\frac{\log(1-x)}{x}dx-\int_{0}^{1}\frac{\log(1-x)}{x}dx\biggr)\\
				&=\frac{\beta}{\pi}\bigl(-\mathrm{Li}_2(-1)+\mathrm{Li}_2(1)\bigr) \, ,
			\end{aligned}
		\end{equation*}
		where $\mathrm{Li}_2$ is the dilogarithm (or Spence's) function, which admits a series representation for $|x|\leq 1$ by
		$
		\mathrm{Li}_2(x)=\sum_{k=1}^\infty\frac{x^k}{k^2} \,
		$,
		so that
		\begin{equation*}
			\begin{aligned}
				\int_{0}^{\infty}\log(\frac{e^{\pi t/\beta}+1}{e^{\pi t/\beta}-1})dt=\frac{\beta}{\pi}\frac{3\pi^2}{12}=\frac{\beta\pi}{4},
			\end{aligned}
		\end{equation*}
		because of 
		\begin{equation*}
			-\mathrm{Li}_2(-1)=\sum_{k=1}^\infty\frac{(-1)^{k+1}}{k^2}=\frac{\pi^2}{6}-2\sum_{k=1}^\infty\frac{1}{(2k)^2}=\frac{\pi^2}{12} \, .
		\end{equation*}
		Together this finishes the first part of the proof 
		\begin{equation}\label{eq:simple-upper-bound-case0}
			\begin{aligned}
				\Bigl|&\tr[\rho_{\beta}[H]O_A]-\tr[\rho_{\beta}[H+V_B]O_A]\Bigr|\leq 2\beta \|O_A\|\,\|V_B\|\,.
			\end{aligned}
		\end{equation}

		\noindent \underline{\textbf{Case II: $d(O_A,V_B)\geq1$}}.             Similarly, we apply the identity \cite[Thm.~1]{Kliesch.2014}
		\begin{equation}\label{eq:proof-integral-equation}
			\Bigl|\tr[\rho_{\beta}[H]O_A]-\tr[\rho_{\beta}[H+V_B]O_A]\Bigr|=\Bigl|\beta\iint_0^1\mathrm{Cov}_{\rho_\beta[H(s)]}^\tau(O_A,V_B)ds d\tau\Bigr|,
		\end{equation}
		
		where
		
		\begin{equation}\mathrm{Cov}_\rho^{\tau}(O_A,O_B) = \mathrm{Tr}\!\left( \rho^{\tau} O_A \rho^{1-\tau} O_B \right) - \mathrm{Tr}\!\left( \rho O_A \right)\mathrm{Tr}\!\left( \rho O_B \right),
		\end{equation}
		
		and then relate the generalized covariance to the standard covariance. Here, we use the upper bound proven in \cite[Eq.~C.22]{Kuwahara.2022}, so that for all $\tau \in [0,1]$ (see also \cite{Hauke_2016}) 
		\begin{equation*}
			\left|\mathrm{Cov}_{\rho_\beta[H(s)]}^\tau(O_A, V_B) - \mathrm{Cov}_{\rho_\beta[H(s)]}(O_A, V_B)\right| \leq \frac{2}{\beta} \int_{-\infty}^{\infty} \frac{e^{-2\pi|t|/\beta}}{1 - e^{-2\pi|t|/\beta}} \|[O_A(t), V_B]\| dt,
		\end{equation*}
		where $O_A(t) \coloneqq e^{itH(s)} O_A e^{-itH(s)}$. Next, the long-range Lieb-Robinson bound assumed in the statement (see Eq.~(\ref{eq:def-lieb-robinson-bound-prelim})) for $b= \frac{d(A,B)}{2v}$ with $v=2gu$ shows that
		\begin{equation}\label{eq:proof-apply-lrb}
			\begin{aligned}
				|\mathrm{Cov}_{\rho_\beta[H(s)]}^\tau&(O_A, V_B) - \mathrm{Cov}_{\rho_\beta[H(s)]}(O_A, V_B)|/(\|O_A\| \|V_B\|) \\
				&\leq \frac{2}{\beta}\int_{-\infty}^{\infty} \frac{e^{-2\pi|t|/\beta}}{1 - e^{-2\pi|t|/\beta}} 
				\begin{cases}
					\kappa_{LR} |A| |B|\frac{g(t)}{(1+d(A,B))^{\alpha}}, & \text{if } b \geq |t|, \\
					2, & \text{otherwise}
				\end{cases} dt \\
				&\leq \frac{4}{\beta}\biggl( \frac{\kappa_{LR}|A| |B|}{\bigl(1+d(A,B)\bigr)^{\alpha}} \underbrace{\int_{0}^{b} \frac{e^{-2\pi t/\beta}g(t)}{1 - e^{-2\pi t /\beta}} dt}_{\mathbf{I.1}} + 2\underbrace{\int_{b}^{\infty} \frac{e^{-2\pi t /\beta}}{1 - e^{-2\pi t /\beta}} dt}_{\mathbf{I.2}} \biggr).
			\end{aligned}
		\end{equation}
		By assumption, we can directly apply $\mathbf{I.1}\leq I_c$ and continue with substituting $x = e^{-2\pi t/ \beta}$ in $\mathbf{I.2}$:
		\begin{equation*}
			\frac{2\pi}{\beta} \int_{b}^{\infty} \frac{e^{-2\pi t /\beta}}{1 - e^{-2\pi t /\beta}} dt= -\int_{e^{-2\pi b /\beta}}^{0} \frac{1}{1 - x} dx= \ln\biggl(1 + \frac{e^{-2\pi b /\beta}}{1 - e^{-2\pi b /\beta}}\biggr)\,.
		\end{equation*}
		Then, 
		\begin{equation}\label{eq:upper-bound-constant-integral_TI}
			\begin{aligned}
				\int_{b}^{\infty} \frac{e^{-2\pi t /\beta}}{1 - e^{-2\pi t /\beta}} dt&\leq \frac{\beta}{2\pi}\,\frac{e^{-2\pi b /\beta}}{1-e^{-2\pi b /\beta}}=\frac{\beta}{2\pi}e^{-2\pi b /\beta}\Bigl(1+\frac{1}{e^{2\pi b /\beta}-1}\Bigr)\leq\frac{\beta}{2\pi}e^{-2\pi b /\beta}\Bigl(1+\frac{\beta}{2\pi b}\Bigr)
			\end{aligned}
		\end{equation}
		converging exponentially with the distance $d(A,B)$ ($b=\tilde{b}d(A,B)$) and temperature $\beta^{-1}$. Next, we upper bound Eq.~(\ref{eq:proof-apply-lrb}) with the help of the bounds on $\mathbf{I.1}$ and $\mathbf{I.2}$, which shows
		\begin{equation}
			\begin{aligned}
				|\mathrm{Cov}_{\rho_\beta[H(s)]}^\tau&(O_A, V_B) - \mathrm{Cov}_{\rho_\beta[H(s)]}(O_A, V_B)|/(\|O_A\| \|V_B\|) \\
				&\leq \frac{4}{\beta}\biggl( \frac{\kappa_{LR}|A| |B|}{\bigl(1+d(A,B)\bigr)^{\alpha}} I_c(\beta) +  \frac{\beta}{\pi} e^{-\frac{2\pi d(A,B)}{\beta v}} \left( 1+\frac{v\beta}{2\pi d(A,B) }\right)  \biggr)\\
				&\leq \frac{c_3 \left(\beta\right)}{\beta} |A||B|  \frac{1}{\left(1+d(A,B)\right)^\alpha} ,
				\label{eq:proof-final-bound-cov}
			\end{aligned}
		\end{equation}
		where $c_3 \left(\beta\right)$ is polynomial in $\beta$. Finally, perturbing the long-range Hamiltonian locally by $V$ supported on $B \subset \Lambda$ for $d(A,B)\geq1$ can be upper bounded using \eqref{eq:proof-integral-equation} combined with the bound \eqref{eq:proof-final-bound-cov} and the assumption of the decay of correlations (see Def.~\ref{def:decay-correlation}):
		\begin{equation}
			\begin{aligned}
				|\tr[O_A \rho_{\beta[H]}] - \tr[O_A \rho_{\beta[H+V]}]| &\leq \Bigl|\beta\iint_0^1\mathrm{Cov}_{\rho_\beta[H(s)]}^\tau(O_A,V_B)ds d\tau\Bigr|\\
				&\leq c_3 \left(\beta\right)\|O_A\| \|V_B\| |A||B| \frac{1}{\left(1+d(A,B)\right)^\alpha}+ \Bigl|\beta\int_0^1\mathrm{Cov}_{\rho_\beta[H(s)]}(O_A,V_B)ds\Bigr|\\
				&\leq c_3 \left(\beta\right) \|O_A\| \|V_B\||A||B| \frac{1}{\left(1+d(A,B)\right)^\alpha}+ \beta C \|O_A\| \|V_B\| \frac{|A|\,|B| e^{(|A|+|B|)/k}}{(1+d(A, B))^\alpha}\\
				&\leq \kappa \left(\beta\right) \|O_A\| \|V_B\| \frac{|A| |B| e^{(|A|+|B|)/k}}{(1+d(A, B))^{\alpha}} \, ,
			\end{aligned}
		\end{equation}
		where $\kappa (\beta)$ is a polynomial in $\beta$ and also in the constants of the Lieb-Robinson bound. Notice that we did not assume the temperatures to be high at this point but only decay of correlations and the Lieb-Robinson bound from Eq.~\eqref{eq:lieb-robinson-bound-appx}.
	\end{proof}
	
	\subsection{From LPPL to local indistinguishability} \label{app:lppltoloc}
	
	\begin{lemma}\label{lem:lppl-local-indist}
		Assume that the Hamiltonians $H$ defined in Equation (\ref{eq:def-long-range}) for a given $\alpha> D$ satisfies the LPPL property (\ref{def:stability-local-perturbation}) for all $V\in\{H_E\,|\,E\subset\Lambda\}$. Then, for all $\delta\in(0,\alpha-D)$
		\begin{equation*}
			|\tr[O_A \rho_\beta[H] b] - \tr[O_A \rho_\beta[H_B]]| \leq \kappa' \zeta(1+\delta) \|O_A\| \frac{|A|^2 e^{|A|/k}}{(1+d(A, C))^{\alpha-D-\delta}},
		\end{equation*}
		where $A \subset B \subset \Lambda$, $H_B = \sum_{Z \subset B } h_Z$, and $\zeta$ the Riemann zeta function, e.g.~$\zeta(2)=\frac{\pi^2}{6}$.
	\end{lemma}
	\begin{proof}
		First, we define $C=\Lambda\backslash B$. In a first step, we number all Hamiltonian terms $h_Z$ with $Z\cap B\neq \emptyset\neq Z\cap C$ (that is, those of $H-H_B$) from $1$ to $M\in\N$, i.e.~$\{h_1,...,h_M\}=\{h_Z\,|\,Z\cap B\neq \emptyset\neq Z\cap C\}$. Let us set $h_0=h_{M+1}=0$. Then, the telescopic sum together with LPPL show
		\begin{equation*}
			\begin{aligned}
				\left|\tr[O_A \rho_\beta[H]] - \tr[O_A \rho_\beta[H_B]]\right|&\leq\sum_{j=0}^{M}\bigg|\tr\Bigl[O_A \rho_\beta \Big[H-\sum_{i=1}^j h_i\Big]\Bigr] - \tr\Bigl[O_A \rho_\beta\Big[H-\sum_{i=1}^{j+1} h_i\Big]\Bigr]\bigg|\\
				&\leq K' \sum_{j=1}^{M} \|O_A\| \|h_{j}\| \frac{|A|k e^{|A|/k+1}}{(1+d(A, Z_j))^{\alpha}} \, ,
			\end{aligned}
		\end{equation*}
		where $Z_j$ is the support of $h_j$ for all $j\in\{1,...,M\}$ and by assumption $|Z_j|\leq k$. Next, we restructure the sum over $1,...,M$ by all $Z$ such that $\{i,j\}\in Z$ with $i\in B$, $j\in C$ and $d(A,Z)=d(A,i)$.
		\begin{equation*}
			\begin{aligned}
				|\tr[O_A \rho_\beta[H]] - \tr[O_A \rho_\beta[H_B]]|
				&\leq K' |A|k e^{|A|/k+1}\|O_A\|\sum_{\substack{Z\,|\,Z\ni\{i,j\}\in B\times C, \\ \quad\,d(A,i)=d(A,Z)}}\frac{\|h_{Z}\|}{(1+d(A, i))^{\alpha}}\\
				&\leq K' |A|k e^{|A|/k+1}g\|O_A\|\sum_{\{i,j\}\in B\times C}\frac{1}{(1+d(A, i))^{\alpha}(1+d(i,j))^{\alpha}}\\
				&\leq K' |A|k e^{|A|/k+1}g\|O_A\|\sum_{j\in C}\frac{  u_B}{(1+d(A, j))^{\alpha}}\\
			\end{aligned}
		\end{equation*}
		where we used \eqref{eq:def-long-range} in the second inequality and applied Lem.~\ref{lem:convolution} in the last step as well as the fact that if $\ell\in A$ is such that $d(A,i)=d(\ell,i)$, then $(1+ d(\ell,j))^\alpha\leq (1+ d(A,j))^\alpha$. Here, $u_B=\max_{j\in \Lambda}\sum_{i\in B}\frac{2^\alpha}{(1+d(i,j))^\alpha} \le u$. Next, we consider the level sets
		\begin{equation*}
			\partial C_{\ell}^{A}=\{j\in C\,|\,d(A,j)=\ell\},
		\end{equation*}
		which have cardinality upper bounded by 
		\begin{equation}
			|\partial C_\ell^A|\leq |A|2^{D}\binom{D+\ell-1}{\ell}\leq |A|2^D\frac{(D+\ell-1)^{D-1}}{(D-1)!}\leq 2|A|e^2(D+\ell-1)^{D-1}
		\end{equation}
		due to \eqref{eq:surface-bound}. Therefore,
		\begin{equation*}
			\begin{aligned}
				|\tr[O_A \rho_\beta[H]] - \tr[O_A \rho_\beta[H_B]]|
				&\leq 2 K' e^3 |A|^2k e^{|A|/k}g u_B \|O_A\|\sum_{\ell =d(A,C)+1}^\infty\frac{(D+\ell-2)^{D-1}}{\ell^{\alpha}}\,.
			\end{aligned}
		\end{equation*}
		Next, we upper bound the series.
		The case $D=1$ is trivial because
		\begin{equation*}
			\begin{aligned}
				\sum_{\ell =d(A,C)+1}^\infty\frac{1}{\ell^{\alpha}}
				&\leq\frac{1}{(d(A,C)+1)^{\alpha-1-\delta}}\sum_{\ell=1}^\infty\frac{1}{\ell^{1+\delta}}\\
				&\leq\frac{1}{(d(A,C)+1)^{\alpha-1-\delta}}\zeta(1+\delta)\\
				&\overset{\delta=1}{=}\frac{1}{(d(A,C)+1)^{\alpha-2}}\frac{\pi^2}{6}\,.
			\end{aligned}
		\end{equation*}
		Note that the last line is only valid if $\alpha-D>1$, which allows for the choice $\delta=1$. Otherwise, we stop the argument one line before.
		The case $D\geq2$ is bounded as follows
		\begin{equation*}
			\begin{aligned}
				\sum_{\ell =d(A,C)+1}^\infty\frac{(D+\ell-2)^{D-1}}{\ell^{\alpha}}&=(D-1)^{D-1}\sum_{\ell =d(A,C)+1}^\infty\frac{(1+\frac{\ell-1}{D-1})^{D-1}}{\ell^{\alpha}}\\
				&\leq(D-1)^{D-1}\sum_{\ell =d(A,C)+1}^\infty\frac{1}{\ell^{\alpha-D+1}}\\
				&\leq\frac{1}{(d(A,C)+1)^{\alpha-D-\delta}}\sum_{\ell=1}^\infty\frac{1}{\ell^{1+\delta}}\\
				&\leq\frac{1}{(d(A,C)+1)^{\alpha-D-\delta}}\zeta(1+\delta)\\
				&\overset{\delta=1}{=}\frac{1}{(d(A,C)+1)^{\alpha-D-1}}\frac{\pi^2}{6} \, ,
			\end{aligned}
		\end{equation*}
		so that
		\begin{equation*}
			\begin{aligned}
				|\tr[O_A \rho_\beta[H]] - \tr[O_A \rho_\beta[H_B]]|
				&\leq 2 K' e^3 |A|^2k e^{|A|/k}g u \|O_A\|\frac{1}{(d(A,C)+1)^{\alpha-D-\delta}}\zeta(1+\delta) \, .
			\end{aligned}
		\end{equation*}
		Denoting $\kappa' = 2 K' e^3 kg u$, we conclude.
	\end{proof}
	
	\subsection{From local indistinguishability to decay of correlations}\label{app:loctodec}
	
	\begin{figure}[t!]
		\begin{center}
			\includegraphics[scale=0.25]{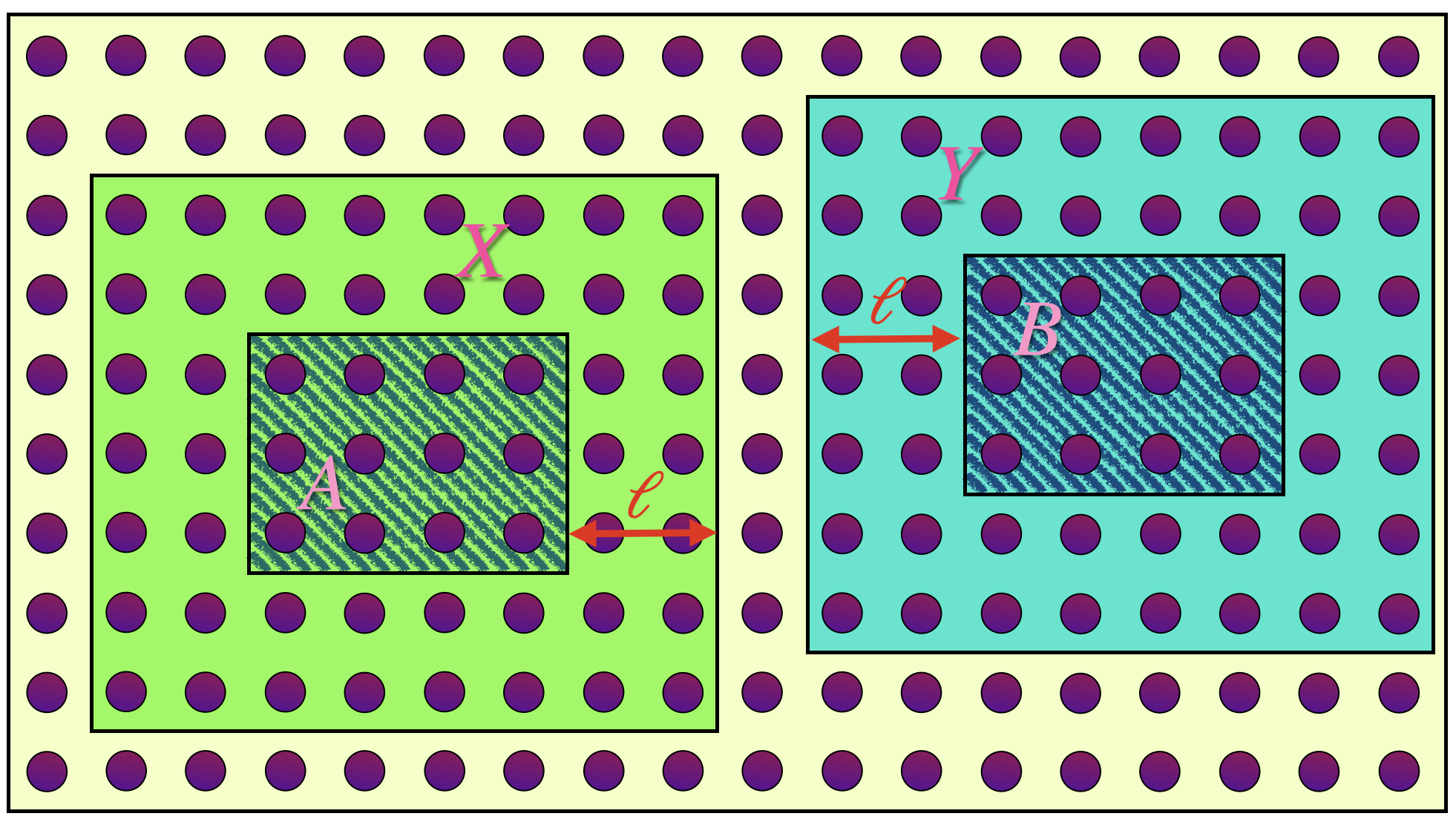}
			\caption{Subsets $A\subset X\subset \Lambda$ and $B\subset Y\subset \Lambda$ so that $d(A,X^c)=\ell =d(B,Y^c)$, for $\ell=3$, and $2\ell\leq d(A,B)$. }
			\label{fig5}
		\end{center}
	\end{figure}  
	
	\begin{lemma}\label{lem:local-indist-decay-cor}
		Assume that the Hamiltonians $H$ defined in Equation (\ref{eq:def-long-range}) satisfy the local indistinguishability (\ref{def:local-indistinguishability}) for a given $\alpha>2D$. Then,
		\begin{equation*}
			\begin{aligned}
				\mathrm{Cov}_{\rho_\beta}(O_A, O_B)
				&\leq \kappa''\|O_A\|\,\|O_B\|\,|A|\,|B|e^{(|A|+|B|)/k}\frac{1}{(1+d(A,B))^{\alpha-2D}}
			\end{aligned}
		\end{equation*}
		where $O_A\in\cB(\cH_A), O_B\in\cB(\cH_B)$ and $\kappa''\geq 0$ a constant.
	\end{lemma}
	\begin{proof}
		We follow the proof of \cite[Theorem 21]{capel2023lppl} and extend the steps to long-range systems. The case $d(A,B)\leq3$ is direct by Hölder's inequality. Assuming that $d(A,B)\geq4$, there are subsets $A\subset X\subset \Lambda$ and $B\subset Y\subset \Lambda$ so that $d(A,X^c)=\ell\geq1$, $d(B,Y^c)=\ell\geq1$ and $2\ell\leq d(A,B)$ (see \cref{fig5}). Then, the assumption of local indistinguishability shows
		\begin{equation}\label{eq:proof-step-local-indist}
			\begin{aligned}
				|\tr[\rho_\beta[H]O_AO_B]-\tr[\rho_\beta[H_{X\cup Y}]O_AO_B]|&\leq K'' \|O_A\|\|O_B\| \frac{(|A|+|B|) e^{(|A|+|B|)/k}}{(1+\min\{d(A, X^c),d(B, Y^c)\})^{\alpha}},\\
				|\tr[\rho_\beta[H]O_A]-\tr[\rho_\beta[H_{X}]O_A]|&\leq K'' \|O_A\| \frac{|A| e^{|A|/k}}{(1+d(A, X^c))^{\alpha}},\\
				|\tr[\rho_\beta[H]O_B]-\tr[\rho_\beta[H_{Y}]O_B]|&\leq K'' \|O_B\| \frac{|B| e^{|B|/k}}{(1+d(B, Y^c))^{\alpha}}\,.
			\end{aligned}
		\end{equation}
		To combine the first with the last equation, we approximate $\rho_\beta[H_{X\cup Y}]$ by $\rho_\beta[H_{X}]\otimes\rho_\beta[H_{Y}]$ in the next step. For that, we identify the interactions of the Hamiltonian between $X$ and $Y$. These are given by
		\begin{equation*}
			\begin{aligned}
				\|H_{X\cup Y}-H_{X}-H_{Y}\|&\leq\sum_{\{j,j'\}\in X\times Y}\sum_{Z : \{j,j'\}\in Z}\|h_Z\|\\
				&\leq\sum_{\{j,j'\}\in X\times Y}\frac{g}{(1+d(j,j'))^\alpha}\\
				&\leq\frac{g |X||Y|}{(1+d(X,Y'))^\alpha}
			\end{aligned}
		\end{equation*}
		which implies by Eq.~\eqref{eq:simple-upper-bound-case0}
		\begin{equation}\label{eq:proof-step-tensor-approx}
			\begin{aligned}
				|\tr[\rho_{\beta}[H_{X\cup Y}]O_R]-\tr[\rho_{\beta}[H_{X}]\otimes\rho_{\beta}[H_{X}]O_R]|&\leq2\beta \|O_R\|\,\|H_{X\cup Y}-H_{X}-H_{Y}\|\\
				&\leq2\beta \|O_R\|\,\frac{g |X||Y|}{(1+d(X,Y))^\alpha}
			\end{aligned}
		\end{equation}
		for any $R\subset \Lambda$. Combining the bounds Eq.~\eqref{eq:proof-step-local-indist} and (\ref{eq:proof-step-tensor-approx}) proves
		\begin{equation*}
			\begin{aligned}
				\mathrm{Cov}_{\rho_\beta}&(O_A, O_B)\\
				&=\tr[\rho_\beta[H]O_AO_B]-\tr[\rho_\beta[H]O_A]\tr[\rho_\beta[H]O_B]\\
				&\leq\tr[\rho_\beta[H_{X\cup Y}]O_AO_B]-\tr[\rho_\beta[H_X]O_A]\tr[\rho_\beta[H_Y]O_B]\\
				&\qquad+K'' \|O_A\|\|O_B\|\Bigl( \frac{(|A|+|B|) e^{(|A|+|B|)/k}}{(1+\min\{d(A, X^c),d(B, Y^c)\})^{\alpha}}+\frac{|A| e^{|A|/k}}{(1+d(A, X^c))^{\alpha}}+\frac{|B| e^{|B|/k}}{(1+d(B, Y^c))^{\alpha}}\Bigr)\\
				&\leq \|O_A\|\|O_B\|\,\Bigl( \frac{ 2 \beta g |X||Y|}{(1+d(X,Y))^\alpha}+K''\frac{(|A|+|B|) e^{(|A|+|B|)/k}+|A| e^{|A|/k}+|B| e^{|B|/k}}{(1+\min\{d(A, X^c),d(B, Y^c)\})^{\alpha}}\Bigr)\\
				&\leq \|O_A\|\|O_B\|\,\Bigl( \frac{ 2 \beta g |X||Y|}{(1+d(X,Y))^\alpha}+K''\frac{4|A||B| e^{(|A|+|B|)/k}}{(1+\min\{d(A, X^c),d(B, Y^c)\})^{\alpha}}\Bigr)\\
				&\leq \|O_A\|\,\|O_B\|\,|A|\,|B|e^{(|A|+|B|)/k} ( 2 \beta g+4 K'')\,\Bigl(\frac{\ell^{2D}}{(1+d(A,B)-2\ell)^\alpha}+\frac{1}{(1+\ell)^{\alpha}}\Bigr)\,,
			\end{aligned}
		\end{equation*}
		where, we choose $X=\{x\in\Lambda\,|\,d(A,x)\leq \ell\}$ and $Y=\{y\in\Lambda\,|\,d(B,y)\leq \ell\}$ for an $\ell\in\N$ such that $2\ell\leq d(A,B)$ and by triangle inequality
		\begin{equation*}
			d(A,B)\leq \min_{cx\in X, y\in Y}\bigl(d(A,x)+d(x,y)+d(y,B)\bigr)\leq d(X,Y)+2\ell
		\end{equation*}
		for all $x\in X$ and $y\in Y$. Moreover, the sets $X,Y$ satisfy $|X|\leq |A|\ell^{D}$ and $|Y|\leq |B|\ell^{D}$ assuming that the lattice is hypercubic. Next, we choose $\ell=\lfloor\frac{1}{3}d(A,B)\rfloor$, so in particular $\ell\leq\frac{1}{3}d(A,B)$ which shows
		\begin{equation*}
			\begin{aligned}
				\mathrm{Cov}_{\rho_\beta}(O_A, O_B)
				&\leq \|O_A\|\,\|O_B\|\,|A|\,|B|e^{(|A|+|B|)/k} (2 \beta g+4K'')\,\Bigl(\frac{\ell^{2D}}{(1+d(A,B)-2\ell)^\alpha}+\frac{1}{(1+\ell)^{\alpha}}\Bigr)\\
				&\leq \|O_A\|\,\|O_B\|\,|A|\,|B|e^{(|A|+|B|)/k} (2 \beta g+4K'')\,3^{\alpha-2D}2\frac{1}{(1+d(A,B))^{\alpha-2D}}
			\end{aligned}
		\end{equation*}
		with the help of 
		\begin{equation*}
			\begin{aligned}
				\frac{\ell^{2D}}{(1+d(A,B)-2\ell)^\alpha}\leq\frac{\bigl(\frac{d(A,B)}{3}\bigr)^{2D}}{(1+\frac{1}{3}d(A,B))^\alpha}\leq3^{\alpha-2D}\frac{(1+d(A,B))^{2D}}{(1+d(A,B))^\alpha}
			\end{aligned}
		\end{equation*}
		which finishes the proof taking $\kappa''=(2 \beta g+4K'')3^{\alpha-2D}2$.
	\end{proof}
	
	\section{Auxiliary results}
	This appendix contains two auxiliary results that are used in the proofs of the main results of the paper. We start with the first auxiliary result on the convolution of long-range interactions:
	\begin{lemma}[Lemma 2 of \cite{kim2024arealawslongrange}]\label{lem:convolution}
		For any $\alpha >D$,
		\begin{equation}
			\underset{j \in \Lambda}{\sum} \frac{1}{(1+d_{i,j})^\alpha}\frac{1}{(1+d_{j,k})^\alpha} \leq \frac{   u}{(1+d_{i,k})^\alpha} \, ,
		\end{equation}
		where $u = \underset{i \in \Lambda}{\operatorname{sup}}\underset{j\in \Lambda}{\sum} \frac{2^\alpha}{(1 + d_{i,j})^\alpha}$.
	\end{lemma}
	\begin{proof}
		A similar property has already been used in previous works such as \cite{Nachtergaele_2006,hastings2010quasiadiabatic}. We omit the proof of this lemma here, as it can be found e.g.~in  \cite{kim2024arealawslongrange}.
	\end{proof}
	
	Next, we continue with a few facts on surface cardinalities: The surface of a ball in graph norm with diameter $\ell$ defined in a $D$-dimensional hypercubic lattice is calculated as follows. First, we define the lattice via an orthonormal basis $\{v_1,...,v_D\}\subset \C^D$ by
	\begin{equation*}
		\Lambda= \{\sum_{j=1}^{D}a_jv_j\,|\,\text{for }a\in\Z^D\}\,.
	\end{equation*}
	Then, the ball (in graph norm) of diameter $\ell$ is defined by
	\begin{equation*}
		B_\ell=\Bigl\{\sum_{j=1}^{D}a_jv_j\,|\,\sum|a_j|\leq\ell\text{ for all }a\in\Z^D\Bigr\}\,,
	\end{equation*}    
	so its surface is
	\begin{equation*}
		\partial B_\ell=\Bigl\{\sum_{j=1}^{D}a_jv_j\,|\,\sum_{j=1}^D|a_j|=\ell\text{ for all }a\in\Z^D\Bigr\}\,.
	\end{equation*}
	The cardinality of the surface can be translated to a ball in marked bins problems. The idea is to distribute $\ell$ identical balls to $D$ boxes. Since each bin could contain no ball, a negative number of balls or a positive number of balls due to the absolute value, we further reduce the problem to
	\begin{equation*}
		\partial B_\ell^+=\Bigl\{\sum_{j=1}^{D}a_jv_j\,|\,\sum_{j=1}^D|a_j|=\ell\text{ for all }a\in\N_{\geq0}^D\Bigr\}\,,
	\end{equation*}
	which suppresses the possible negative values. As mentioned, the cardinality of the above set is equal to the number of possibilities to distribute $\ell$ identical balls to $D$ bins, which is given by \cite[Chap.~1.9]{Stanley.1986combinatorics} 
	\begin{equation*}
		|\partial B_\ell^+|=\binom{D+\ell-1}{\ell}\,.
	\end{equation*}
	By definition,
	\begin{equation*}
		\Bigl|\partial B_\ell\Bigr|\leq\Bigl|\Bigl\{\sum_{j=1}^D|n_ja_j|=\ell\text{ for all }a\in\N_{\geq0}^D\wedge n\in\{-1,1\}^D\Bigr\}\Bigr|=2^D|\partial B_\ell^+|\,.
	\end{equation*}
	This is only an upper bound because it overcounts the elements $(-1)\cdot 0$. This can be further upper bounded by
	\begin{equation}\label{eq:surface-bound}
		|\partial B_\ell|\leq 2^{D}\binom{D+\ell-1}{\ell}\leq 2^D\frac{(D+\ell-1)^{D-1}}{(D-1)!}\leq 2e^2(D+\ell-1)^{D-1}\,.
	\end{equation}
	
	\section{Analysis of numerical results}\label{app:sim}
	
	\label{app:num}
	\subsection{Further results and discussion of the high-temperature phase }
	
	\begin{figure}
		\centering
		\includegraphics[width=1.0\linewidth]{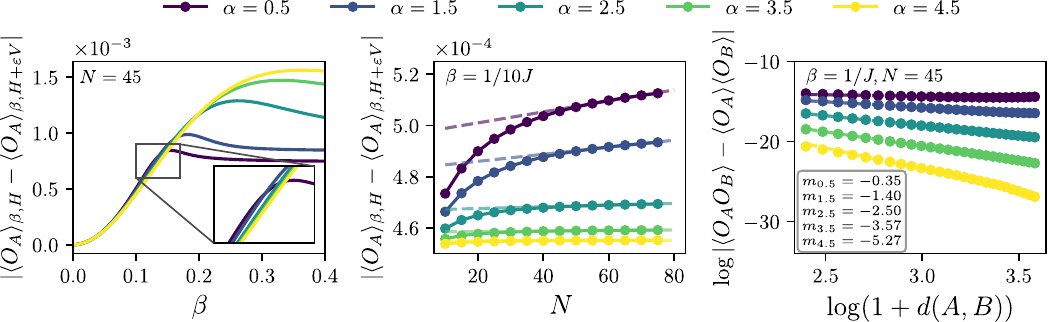}
		\caption{\justifying Absolute value of the difference in the expectation value of the observable $O_A$ between perturbed and unperturbed Gibbs states of Hamiltonian given by Eq.~\eqref{eq::HLRising} against $\beta$ for $N=45$ (left) and system size $N$ (middle), which contrast with Fig.~\ref{fig:DifvsEverything} due to the variation of $\beta$. This figure showcases a change of dependence on $\alpha$ when $\beta$ increases. We also plot the logarithm of the Gibbs state covariance varying the distance of the observables $O_A$ and $O_B$ (right). We display the slopes of the best linear fit for the asymptotic behavior. The $r$ coefficient of the linear fits is approximately $r \approx -0.99$ for all cases but for $r_{\alpha=0.5}\approx-0.9$.}
		\label{fig:appendix}
	\end{figure}
	
	In Fig.~\ref{fig:appendix} we provide additional plots to complement those of the main text. 
	First, we explore the temperature dependence of the quantities analyzed. 
	In Fig.~\ref{fig:DifvsEverything}, we see that the difference in expectation values increases with $\alpha$. However,
	we observe a change of that dependence on $\alpha$ for $|\langle A \rangle_{\beta,H} - \langle A \rangle_{\beta,H+\varepsilon V}|$ when lowering $\beta$, as shown in Fig.~\ref{fig:appendix} (left-most plot). While this is likely dependent on the model at hand, it is also compatible with our bound due to the two terms in Eq.~\eqref{eq:proof-final-bound-cov}. The first admits a polynomial decay in the distance $d(A,B)$ of degree $\alpha$, and the second an exponential decay in $d(A,B)/\beta$ dependent on the inverse temperature, which naturally allows for different behavior at small versus large $\beta$ and accounts for the observed monotonic behavior, but also permits a change in the ordering of the curves.
	
	We also analyze the polynomial decay of correlations with a log-log plot in Fig.~\ref{fig:appendix}. We observe that, asymptotically, in the weak long-range regime $\alpha >D=1$, the correlations decay spatially at least as $ \propto d^{-\alpha}$ as expected, with the slight exception of $\alpha=1.5$, which decays as $ \propto d^{-1.4}$, possibly due to finite-size effects.
	
	At high temperatures, the Gibbs state approximately approaches a completely random state, $\rho_\beta \approx I/d$, which implies that in this regime, the system is expected to be in a disordered phase where correlations are nearly negligible and decay rapidly with distance. This may provide some intuition as to why correlations decay with distance significantly faster than analytical bounds predict, even at finite temperatures. Combined with the fact that the observables $O_A = \langle S^z_i S^z_{i+1} \rangle$ are of the order $\mathcal{O}(10^{-3})$, this suggests that the system remains in a disordered phase within that regime of temperature.
	
	\subsection{Simulation details}
	The numerical simulations consist of a TEBD-like algorithm where we construct the time evolution operator of the LR-TFI model Eq.~\eqref{eq::HLRising} explicitly in the short-range case~\cite{Pirvu2010, MPOTEv_Paeckel_2019} and then apply a super-extensive amount of swap gates to exactly simulate the power-law interactions in a finite system size.
	This yields a Trotterized representation of the imaginary time evolution operator $U(\delta \beta) = \mathrm{e}^{H \delta \beta}$ in form of an MPO. For it, we have used the ITensor library in Julia~\cite{Fishman2022}.
	We concatenate the Trotterized imaginary time evolution to yield the desired $\beta$ as $ U(\beta) \simeq \prod_{i=1}^{n_\beta} U(\delta\beta)$, with $\beta = n_\beta \cdot \delta \beta$. Once we have the Gibbs state MPO, we compute expectation values. More concretely the observable in Fig.~\ref{fig:DifvsEverything} and Fig.~\ref{fig:locindist} is given by
	\begin{equation}
		O_A= \sigma_i^z \sigma_{i+1}^z,
	\end{equation}
	where $i$ is the site at the middle of the system.
	
	In the simulation, the relevant parameters are the time-step $\delta \beta= 0.001$ and the SVD cut-off$=10^{-19}$, which have been determined in order not to saturate a maximal bond dimension of $\chi_\mathrm{max} = 1800$ during the time evolution. The parameters of the Hamiltonian in Eq.~\ref{eq::HLRising} we choose are $h=0.25, J=1$.
\end{document}